\documentclass[a4paper,fleqn]{cas-sc}

\usepackage[numbers,sort&compress]{natbib}

\def\tsc#1{\csdef{#1}{\textsc{\lowercase{#1}}\xspace}}
\tsc{WGM}
\tsc{QE}
\tsc{EP}
\tsc{PMS}
\tsc{BEC}
\tsc{DE}

\usepackage{amsmath,amssymb,bm}
\usepackage{float}
\usepackage{multicol}

\usepackage{algorithm}
\usepackage[noend]{algpseudocode}
\usepackage{amsthm}

\newtheorem{thm}{Theorem}[section]
\newtheorem{prop}[thm]{Proposition}

\newtheorem{theorem}[thm]{Theorem}

\usepackage{array,tabularx,booktabs}   
\usepackage{enumitem}

\newcolumntype{Y}[1]{>{\raggedright\arraybackslash}p{#1}}

\definecolor{respblue}{RGB}{0,90,200}

\usepackage{lineno}

\begin{document}
\let\WriteBookmarks\relax
\def\floatpagepagefraction{1}
\def\textpagefraction{.001}

\shorttitle{}

\shortauthors{R. Jin and X. Sun}

\title [mode = title]{Optimal Uncertainty Quantification under General Moment Constraints on Input Subdomains}

\makeatletter
\renewcommand\printorcid[1]{}
\makeatother

\author[1]{Rong Jin}
\author[1]{Xingsheng Sun}
\cormark[1]            

\cortext[1]{Corresponding author at: Department of Mechanical and Aerospace Engineering, University of Kentucky, Lexington, KY 40506, United States.
\textit{E-mail address}: \href{mailto:xingsheng.sun@uky.edu}{xingsheng.sun@uky.edu} (X.~Sun)}

\affiliation[1]{organization={Department of Mechanical and Aerospace Engineering, University of Kentucky},
    city={Lexington},
    state={KY},
    postcode={40506},
    country={United States}}


\begin{abstract}
We present an optimal uncertainty quantification (OUQ) framework for systems whose uncertain inputs are statistically independent and characterized by truncated moment constraints defined over subdomains. Given such partial uncertainty information, rigorous optimal upper and lower bounds on the probability of failure (PoF) can be derived over the admissible set of probability measures, providing a principled basis for system safety certification. We formulate the OUQ problem under general moment constraints on subdomains and develop a high-performance computational framework to compute the corresponding optimal bounds. The proposed approach transforms the original infinite-dimensional optimization problems, constrained by subdomain-based moment information, into finite-dimensional unconstrained optimization problems parameterized solely by free canonical moments. To address the prohibitive cost of PoF evaluation in high-dimensional settings, we further incorporate inverse transform sampling (ITS), which enables efficient and accurate estimation of the PoF within the OUQ optimization. We also examine a special case in which the uncertain inputs are constrained only by zeroth-order moments over subdomains, and demonstrate that the resulting OUQ formulation is equivalent to the evidence theory. The effectiveness and scalability of the framework are demonstrated through a set of numerical examples ranging from one-dimensional baselines to a ten-dimensional ballistic impact problem, involving varying numbers of subdomains and moment constraints. The results show that increasing either the number of subdomains or the order of moment constraints systematically decreases the upper bound and increases the lower bound with sufficient uncertainty information, thereby tightening the bound interval. For high-dimensional problems, the ITS strategy reduces the computational cost by up to two orders of magnitude while maintaining a relative error below 1\%. For extreme rare-event settings, the main practical value lies in progressively tightening rigorous upper bounds, which can serve as conservative safety-design criteria. In addition, we identify regimes in which the optimal bounds are more sensitive to subdomain partitioning or to higher-order moment information. These findings provide practical guidance for prioritizing uncertainty reduction efforts when certifying system safety.
\end{abstract}
\maketitle

\section{Introduction}

Uncertainties are inherent in nearly all engineering analyses and decision-making processes, arising from diverse sources ranging from material inhomogeneity to model simplifications~\cite{kovachki2022multiscale, liu2021hierarchical, jin2026ensemble}. These uncertainties can significantly affect the accuracy and reliability of design predictions. Although the use of appropriately chosen safety factors can mitigate the impact of such uncertainties, the resulting designs and products may fail to achieve the desired balance between performance, efficiency, and economy. Therefore, it is essential to quantify and assess the effects of different sources of uncertainty on engineering solutions to provide decision-makers with reliable predictions and quantitative measures of confidence in these predictions.

The design process typically involves constructing a model that reproduces the behavior of the real structure under specified loading conditions. Consequently, it is relatively straightforward to represent the model’s input parameters and output responses using appropriate uncertainty models. Once these uncertainties are characterized, their influence on relevant performance measures can be systematically analyzed and quantified, enabling designers to identify potential weaknesses and refine the design as necessary to enhance robustness and reliability.

Based on their nature and reducibility, uncertainties are commonly classified into aleatory (inherent variability) and epistemic (lack of knowledge)~\cite{der2009aleatory}. Aleatory uncertainties arise from natural randomness---such as material inhomogeneity, manufacturing tolerances, and environmental fluctuations---and are irreducible, even with improved modeling or measurement. Given sufficient data, these inherent variations can typically be characterized using a probability density function (PDF). The most straightforward approach for propagating aleatory uncertainties involves Monte Carlo-based sampling methods~\cite{hastings1970monte}, often grounded in Bayesian inference~\cite{dempster1968generalization}. A critical step in these approaches is the appropriate selection of PDFs for the uncertain parameters, as this choice can significantly influence the resulting predictions~\cite{fetz2004propagation, zhang2018effect, zhang2018quantification}. 

In contrast, epistemic uncertainties originate from incomplete or imperfect knowledge, including model simplifications, numerical approximations, measurement errors, and limited experimental data. Unlike aleatory uncertainties, they can often be reduced through refined models, additional data, or improved experimental techniques. Traditionally, epistemic uncertainties are represented using non-probabilistic methods. For example, evidence theory (also known as Dempster-Shafer theory) expresses uncertainty by assigning degrees of belief and plausibility to events~\cite{dempster2008upper, shafer2020mathematical}. Probability boxes (p-boxes) define envelopes of possible PDFs constrained by available statistical information such as means and variances~\cite{ferson2004arithmetic}. Interval analysis represents uncertain quantities as bounded ranges, with computations carried out using interval arithmetic to ensure that the true (but unknown) values are guaranteed to lie within the resulting intervals~\cite{weichselberger2000theory}. Fuzzy probability theory employs fuzzy numbers, typically characterized by membership functions, to express a continuum of possible probabilities along with their degrees of plausibility~\cite{buckley2005fuzzy}. Epistemic uncertainty can also be represented probabilistically under the Bayesian interpretation, in which probability quantifies degrees of belief rather than intrinsic randomness~\cite{der2009aleatory}. A comprehensive review of epistemic uncertainty models can be found in Ref.~\cite{beer2013imprecise}. It is worth noting that aleatory and epistemic uncertainties are not always easily distinguishable during the characterization of input variables or the modeling and solution of engineering systems. When both types coexist and interact, they are collectively referred to as polymorphic uncertainties~\cite{gotz2015structural}.

More recently, a distinctive approach known as Optimal Uncertainty Quantification (OUQ) has been developed to compute rigorous bounds on quantities of interest (QoIs), such as the upper and lower limits of the probability of failure (PoF)~\cite{owhadi2013optimal}. Unlike traditional methods, OUQ can incorporate partial information about input (or prior) probability measures without requiring the full specification of a PDF or other explicit uncertainty models. Specifically, OUQ reformulates the infinite-dimensional optimization problem of bounding the PoF into a finite-dimensional problem expressed as a convex combination of Dirac measures~\citep{winkler1979integral, winkler1988extreme}. By optimizing the supports and weights of these Dirac measures, OUQ yields the mathematically sharpest possible bounds on the QoI. In the OUQ framework, partial information of uncertainties---such as variable bounds, means, or higher-order moments---is incorporated as constraints in the optimization problem, ensuring that the computed bounds are both rigorous and consistent with the available knowledge. The OUQ framework has been successfully applied to a wide range of engineering problems, including the design of thermal-hydraulic reactors~\citep{stenger2020optimal}, ballistic impact analysis of aluminum and magnesium alloys~\citep{kamga2014optimal, sun2022uncertainty}, rupture prediction of soft collagenous tissues~\citep{balzani2017method}, and the buckling analysis of wide-flange columns~\citep{miska2022method}.

Most previous OUQ studies have focused on incorporating statistical information defined globally over the entire support of random inputs~\cite{stenger2020optimal, kamga2014optimal, sun2022uncertainty, balzani2017method, miska2022method, miska2025reliability, stenger2021optimization}. However, in practical applications, statistical information can also be available only within specific subdomains or subintervals of the input space~\cite{wang2024subdomain, zhao2022sub, liu2015probability, jiang2007optimization}. In many engineering settings, limited data, modeling uncertainty, or incomplete knowledge of dependence structures make it difficult to specify reliable PDFs over the entire input domain. Instead, uncertainty can be characterized through partial statistical information, such as bounds, low-order moments, or region-wise (e.g., subdomain-specific or bin-based) statistics obtained from experiments, field measurements, or physics-based considerations. Such information is generally insufficient to reconstruct the full distribution, particularly in high-dimensional or rare-event regimes where tail behavior is difficult to infer accurately. This perspective is consistent with moment-based ambiguity sets in distributionally robust optimization~\cite{delage2010distributionally, nie2025distributionally}, partition-based uncertainty descriptions~\cite{estebanperez2022partition}, and truncated moment formulations~\cite{curto1996truncated}. Related forms of local statistical characterization also arise in engineering practice, for example in random-field modeling of spatially varying material  roperties~\cite{vanmarcke1986random, sudret2000stochastic, chen2012characterization} and in conditional random fields for geotechnical and material-property modeling~\cite{wang2016using, ouyang2021patching, feng2025efficient}, where field or experimental measurements are used to condition local variability. In addition, computational simulations can be used to estimate statistical moments, such as skewness and kurtosis, through lower-scale statistical volume element simulations~\citep{yin2008statistical}, sparse approximations in moment-based arbitrary polynomial chaos~\citep{ahlfeld2016samba}, and inverse stochastic analyses based on non-sampling generalized polynomial chaos~\citep{sepahvand2014identification}. Nevertheless, higher-order statistical information beyond bounds and means is often difficult to obtain because of limited data availability and computational cost.

These considerations motivate the present framework, which seeks to compute rigorous bounds on quantities of interest over all admissible probability measures consistent with partial and locally defined information. To address these challenges, the present study extends the OUQ framework to compute optimal bounds on the PoF in cases where partial uncertainty information is available locally only over subdomains of the input space. We aim to mathematically formulate and solve the corresponding optimization problem, and to systematically examine how the number of subdomains and the number of prescribed statistical moments affect the optimal bounds on the PoF. Here, the subdomains are prescribed in the input space rather than inferred from regions of the output space. Thus, the proposed formulation does not require prior knowledge of which output regions are influenced by data. Instead, it addresses situations in which local information is available for specified ranges of input variables, such as measurements, design specifications, operating regimes, or field observations associated with prescribed input intervals.

The remainder of the paper is organized as follows. Section~\ref{sec:method} formulates the OUQ problem under general truncated moment constraints and outlines the corresponding solution strategies. Section~\ref{sec:evidence} examines a special case of OUQ under zeroth-order moment constraints and compares the resulting formulation with evidence theory. Section~\ref{sec:examples} presents five numerical examples, including identity functions of four basic distributions, a five-dimensional nonlinear smooth problem, a two-dimensional non-smooth four-branch problem, an eight-dimensional rare-event roof-truss problem, and a ten-dimensional ballistic impact problem. Finally, Section~\ref{sec:summary} concludes the paper with a summary of key findings and a discussion of future directions.

\section{Methodology}
\label{sec:method}

In this section, we begin with a brief overview of the OUQ framework. We then formulate the OUQ problem under general truncated moment constraints and describe its solution using reduction theorem, a canonical-moment-based approach, eigendescompsition of the Jacobi matrix, and inverse transform sampling (ITS). Readers are referred to Ref.~\cite{owhadi2013optimal} for details on the reduction theorem and to Ref.~\cite{stenger2020optimal2} for the canonical-moment methodology.

\subsection{Optimal Uncertainty Quantification and Certification}

In this work, we consider a system characterized by a validated physics-based model. In the context of UQ, we aim to determine the probabilities of system outcomes when the response is stochastic or uncertain due to variability or lack of knowledge in the system inputs or operating conditions. Specifically, the system performance is described by a known response function
\begin{equation}
    Y = G(X) ,
\end{equation}
where the forward mapping \( G: \mathcal{X} \to \mathcal{Y} \) transforms the input space \(\mathcal{X} \equiv \prod_{i=1}^{m} \mathcal{X}_i \subseteq \mathbb{R}^m\) into the output space \(\mathcal{Y} \equiv \prod_{i=1}^{n} \mathcal{Y}_i \subseteq \mathbb{R}^n\). The input vector \( X \equiv (X_1, \dots, X_m) \) consists of \( m \) independent random variables, where each \( X_i \in \mathcal{X}_i \) (\( i = 1, 2, \dots, m \)), representing imperfectly known or uncertain properties of the system. This independence assumption is adopted throughout Sections~\ref{sec:method}-\ref{sec:examples} and is required by the specific admissible-set construction and computational framework developed in the present paper. The output vector \( Y \equiv (Y_1, \dots, Y_n) \), with each \( Y_i \in \mathcal{Y}_i \) (\( i = 1,2, \dots, n \)), corresponds to the system’s performance measures. The random inputs are generated according to an unknown probability law \(\mathbb{P} \in \Phi(\mathcal{X})\), defined on the measurable space \(\mathcal{X}\) and possibly subject to certain known characteristics (e.g., bounds, moments). 

A representative example of uncertainty quantification in engineering is design certification, which involves assessing the probability that a system will perform safely and remain within specified operational limits. The design requirements impose that the output \( Y \) must remain within an \emph{admissible} set \(\mathcal{Y}^{\mathrm{a}} \subseteq \mathcal{Y}\). Failure occurs when \( Y \) falls within the \emph{inadmissible} set \(\mathcal{Y}^{\mathrm{c}} = \mathcal{Y} \setminus \mathcal{Y}^{\mathrm{a}}\). Ideally, the probability measure associated with \( Y \) would have its entire support contained 
within the admissible set, i.e.,
\begin{equation}
    \mathbb{P}[Y \in \mathcal{Y}^{\mathrm{a}}] = 1.
\end{equation}
Systems satisfying this condition can be certified with complete certainty. 
However, such absolute assurance of safe performance is often impractical or unattainable, for example, when \(\mathbb{P}\) lacks compact support or when achieving it would be prohibitively costly. In practice, this strict condition is typically relaxed by introducing an acceptable PoF. Let \(\epsilon \in [0, 1]\) denote the maximum tolerable PoF. 
The system is considered \emph{safe} if
\begin{equation}
    \mathbb{P}[Y \in \mathcal{Y}^{\mathrm{c}}] \leq \epsilon,
\end{equation}
and \emph{unsafe} if
\begin{equation}
    \mathbb{P}[Y \in \mathcal{Y}^{\mathrm{c}}] > \epsilon.
\end{equation}

However, due to incomplete information, the exact probability measure of the input random variables \( X \) is generally unknown. To address this, we define a subset
\begin{equation}
    \mathcal{A} \subseteq \big\{ \mu ~|~ \mu \in \Phi(\mathcal{X}) \big\},
\end{equation}
which encodes all available information about the probability measure \(\mathbb{P}\) of \(X\). This information may originate from experimental data, lower-fidelity simulations, or expert judgment. By construction, the true measure \(\mathbb{P}\) is assumed to belong to \(\mathcal{A}\). Within this admissible set, some scenarios may be \emph{safe} (i.e., \(\mu [Y \in \mathcal{Y}^{\mathrm{c}}] \leq \epsilon\)), 
while others may be \emph{unsafe} (i.e., \(\mu [Y \in \mathcal{Y}^{\mathrm{c}}] > \epsilon\)). Given this information set \(\mathcal{A}\), it is possible to determine {\it rigorous} and {\it optimal} upper and lower bounds on the PoF, corresponding to the most and least conservative scenarios consistent with the assumptions. 
These bounds are defined as
\begin{subequations}
\begin{equation}
    U(\mathcal{A}) := \sup_{\mu \in \mathcal{A}} \mu [Y \in \mathcal{Y}^{\mathrm{c}}],
\end{equation}
and
\begin{equation}
    L(\mathcal{A}) := \inf_{\mu \in \mathcal{A}} \mu [Y \in \mathcal{Y}^{\mathrm{c}}].
\end{equation}
\label{eq:bounds}
\end{subequations}

The bounds \( U(\mathcal{A}) \) and \( L(\mathcal{A}) \) defined in Eq.~(\ref{eq:bounds}) provide a rigorous basis for solving the certification problem. Assuming that the information set \(\mathcal{A}\) is valid (i.e., \(\mathbb{P} \in \mathcal{A}\)), the following criteria apply: if \( U(\mathcal{A}) \leq \epsilon \), the system is \emph{certifiably safe}; if \( \epsilon < L(\mathcal{A}) \), the system is \emph{certifiably unsafe}; and if \( L(\mathcal{A}) \leq \epsilon < U(\mathcal{A}) \), the available information is insufficient to determine the system’s safety. This certification logic and its optimality are illustrated in Fig.~\ref{fig:certify}(a). In addition, the tightness of these bounds depends directly on the information contained in \(\mathcal{A}\). For two information sets \(\mathcal{A}_1\) and \(\mathcal{A}_2\), if \(\mathcal{A}_1 \subseteq \mathcal{A}_2\), then \( U(\mathcal{A}_1) \le U(\mathcal{A}_2) \) and \( L(\mathcal{A}_1) \ge L(\mathcal{A}_2) \). In other words, the more information about uncertainties is incorporated into the UQ analysis, the tighter the resulting bounds \( U(\mathcal{A}) \) and \( L(\mathcal{A}) \) become, and the closer they approach the true PoF, as shown in Fig.~\ref{fig:certify}(b). This tightening improves the efficiency of certification but generally increases computational cost,  which reveals a fundamental trade-off between the economy of certification and computational feasibility.

\begin{figure}[pos=htbp]
\centering
\includegraphics[width=0.9\textwidth]{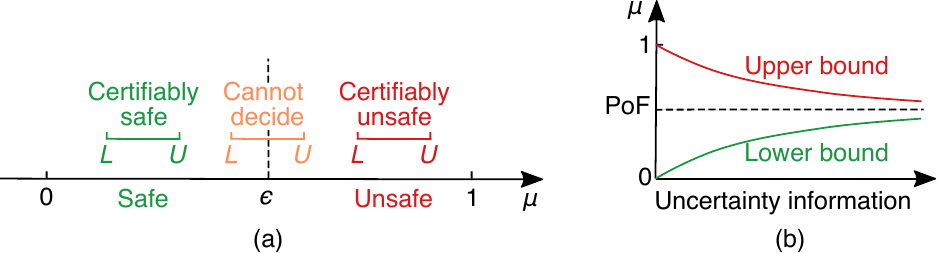}
\caption{Illustration of certification in OUQ: (a) Rigorous decision criterion for system safety, and (b) Convergence of optimal bounds with increasing information.}
\label{fig:certify}
\end{figure}

\subsection{OUQ under Local Moment Constraints}

In this work, we consider a scenario in which the input domain of the system, 
\(\mathcal{X}\), is partitioned into multiple subdomains. Without loss of generality, we assume that the domain of the \(i\)-th input variable, \(\mathcal{X}_i\), is divided into \(K\) subintervals \(\mathcal{X}_{i,j}\) (\(j = 1, \dots, K\)). Each \(\mathcal{X}_{i,j}\) is referred to as a partition of the input domain and satisfies the following properties: (i) the partitions are mutually disjoint,
\begin{equation}
    \mathcal{X}_{i,j} \cap \mathcal{X}_{i,l} = \varnothing, \quad \text{if } j \neq l,
    \label{eq:xixj}
\end{equation}
and (ii) their union covers the entire input space,
\begin{equation}
    \bigcup_{j=1}^{K} \mathcal{X}_{i,j} = \mathcal{X}_i.
    \label{eq:union}
\end{equation}

Due to limited information, the exact probability measure \(\mathbb{P}\) of the inputs may be unknown. However, we assume that the statistical raw moments of the probability measure within each partition can be obtained. Without loss of generality, we further assume that all random inputs share the same number of available moments (denoted by \(r\)) in each subdomain. Accordingly, the information set \(\mathcal{A}\) can be expressed as
\begin{equation}
    \mathcal{A} = 
    \bigg\{
    \mu ~\bigg|~
    \mu = \bigotimes_{i=1}^{m} \mu_i,~
    \int_{\mathcal{X}_{i,j}} x^{k-1} \, \mathrm{d}\mu_i = M_{i,j,k-1},~
    i = 1, 2, \dots, m;~
    j = 1, 2, \dots, K;~
    k = 1, 2, \dots, r+1
    \bigg\},
\label{eq:infset}
\end{equation}
where \(\mu_i\) denotes the probability measure of the \(i\)-th input and \(M_{i,j,k-1}\) represents the \((k-1)\)-th raw truncated moment on the \(j\)-th subinterval of the \(i\)-th input. The product form $\mu = \bigotimes_{i=1}^{m} \mu_i$ is a central assumption of the present formulation. It enables the admissible set to be expressed in terms of marginal subdomain moment constraints and allows the subsequent canonical-moment construction to be carried out separately for each input dimension. The zeroth raw moment corresponds to the probability mass, which 
is constrained by the normalization condition
\begin{equation}
    \sum_{j=1}^{K} M_{i,j,0} = 1.
\end{equation}

For notational simplicity, we assume in this work that all random inputs are partitioned into the same number of subdomains and that each subdomain is characterized by the same number of moment constraints. This assumption is made solely for clarity of presentation and does not restrict the generality of the proposed approach, which can be straightforwardly extended to cases where the number of subdomains or the number of prescribed moment constraints varies among input dimensions.

\subsection{Reduction Theorem}

All realizations in the information set \(\mathcal{A}\) defined in Eq.~(\ref{eq:infset}) are considered equally plausible representations of the underlying reality. Consequently, the most extreme scenarios within \(\mathcal{A}\) are of primary interest for UQ and certification. However, since the optimization problems in Eq.~(\ref{eq:bounds}) are infinite-dimensional due to the uncountably many admissible probability measures \(\mu\), direct solutions are computationally intractable in practice.

To overcome this challenge, we employ the reduction theorem developed in Ref.~\cite{owhadi2013optimal}, which builds upon Winkler’s results on the extreme points of moment sets~\cite{winkler1979integral, winkler1988extreme}. Specifically, the reduction theorem reformulates the infinite-dimensional optimization problem into a finite-dimensional one expressed as a convex combination of Dirac measures. The partial information about uncertainties---namely, the raw moments defined over subdomains in Eq.~(\ref{eq:infset})---is incorporated as constraints in the optimization problem. The theorem states that, for each random input, the number of Dirac supports and corresponding weights equals the number of equality constraints within each subdomain~\cite{owhadi2013optimal}. Consequently, the probability measure can be written as
\begin{equation}
\label{eq:classical_measure_decomposition}
    \mu_i = \sum_{j=1}^{K} \sum_{k=1}^{r+1} t_{i,j,k} \, \delta_{x_{i,j,k}}, 
    \quad i = 1, 2, \dots, m,
\end{equation}
where \(x_{i,j,k}\) and \(t_{i,j,k}\) denote the \(k\)-th support and corresponding weight in the \(j\)-th subdomain of the \(i\)-th input, satisfying
\begin{equation}
    \sum_{j=1}^{K} \sum_{k=1}^{r+1} t_{i,j,k} = 1,
    \quad \text{and} \quad
    x_{i,j,k} \in \mathcal{X}_{i,j}.
\end{equation}

Since there are \(r+1\) moment constraints in each subdomain, the same number of supports and weights are required locally. The computation of the upper and lower bounds \(U(\mathcal{A})\) and \(L(\mathcal{A})\) is thus reduced to solving the finite-dimensional optimization problems
\begin{equation} 
\label{eq:ouq_program_subint}
\begin{aligned}
    U(\mathcal{A})/L(\mathcal{A}) = 
    \underset{\{x\}, \{t\}}{\sup/\inf} \quad &
    \sum_{\alpha_1=1}^{K(r+1)} \!\!\cdots\! 
    \sum_{\alpha_m=1}^{K(r+1)} 
    \left(\prod_{i=1}^m t_{i,\alpha_i}\right) 
    \mathbf{1} \big[G(x_{1,\alpha_1}, \dots, x_{m,\alpha_m}) \in \mathcal{Y}^\text{c}\big], \\[4pt]
    \text{subject to} \quad &
    x_{i,j,k} \in \mathcal{X}_{i,j}, \quad 
    t_{i,j,k} \ge 0, \quad 
    \sum_{l=1}^{r+1} t_{i,j,l} x_{i,j,l}^{k-1} = M_{i,j,k-1}, \\[3pt]
    & i = 1, 2, \dots, m, ~
    j = 1, 2, \dots, K, ~
    k = 1, 2, \dots, r+1,
\end{aligned}
\end{equation}
where \(\{x\} \equiv \{x_{i,j,k}: i= 1, \dots, m;~j = 1, \dots, K;~k = 1, \dots, r+1\} \in \mathbb{R}^{mk(r+1)} \) and \(\{t\} \equiv \{t_{i,j,k}: i= 1, \dots, m,~j = 1, \dots, K,~k = 1, \dots, r+1\} \in \mathbb{R}^{mk(r+1)}\). The product structure in Eq.~\eqref{eq:ouq_program_subint} follows from the independence assumption in Eq.~\eqref{eq:infset}. For dependent inputs, the reduction would need to be reformulated in terms of a joint admissible measure rather than a product of marginal measures. The variables \(x_{i,\alpha_i}\) and \(t_{i,\alpha_i}\) denote the flattened notations enumerating all Dirac supports and weights across subdomains, defined by
\begin{equation}
\label{eq:flattened_index}
x_{i,\alpha_i} = x_{i,j,k}, \quad 
t_{i,\alpha_i} = t_{i,j,k}, \quad 
\text{where } \alpha_i = (r+1)(j-1) + k.
\end{equation}
The flattened notation \(x_{i,\alpha_i}\) and \(t_{i,\alpha_i}\) is adopted in the objective function for conciseness, while the hierarchical form \(x_{i,j,k}\) and \(t_{i,j,k}\) is retained in the constraint expressions. The indicator function \(\mathbf{1}[E]\) is defined as
\begin{equation}
\mathbf{1}[E] =
    \begin{cases}
        1, & \text{if the statement } E \text{ is true}, \\[4pt]
        0, & \text{otherwise.}
    \end{cases}
\end{equation}

\subsection{OUQ Solution }
\label{sec:numerical_method}

Two major challenges arise when solving the constrained optimization problem formulated in Eq.~(\ref{eq:ouq_program_subint}). The first challenge is the discrete nature of the objective function, i.e., the PoF, which is expressed as a sum of indicator functions that take binary values of either $0$ or $1$. This discreteness makes the direct application of gradient-based optimization methods infeasible. The second challenge lies in the enforcement of moment constraints representing partial information about the input variables. This becomes particularly difficult when a large number of equality constraints are imposed, as derivative-free algorithms often struggle to efficiently satisfy them.

To address these challenges, several numerical strategies have been developed. The first is the Mystic framework~\cite{mckerns2012optimal}, which employs a nested (two-level) optimization scheme --- an outer loop that minimizes or maximizes the PoF and an inner loop that identifies feasible candidates that satisfy the imposed constraints. Both loops rely on derivative-free optimizers such as Differential Evolution (DE)~\cite{storn1997differential}. The second approach introduces a neural-network-based approximation of the indicator functions~\cite{sun2023learning}, leveraging their similarity to classification problems in machine learning. This learning-based OUQ formulation enables the use of efficient gradient-based algorithms to solve the resulting optimization problem, wherein the constrained formulation is transformed into an unconstrained one through penalty methods. The third strategy utilizes canonical moments to reformulate the OUQ problem~\cite{stenger2020optimal2}. In this work, since a large number of moment constraints are imposed, both the Mystic framework and penalty-based methods become inefficient. The Mystic framework struggles to sufficiently explore the feasible constrained space within the inner loop, while the penalty-based methods introduce significant numerical errors due to the large penalty terms dominating the objective function. Consequently, we employ the canonical-moment-based approach to efficiently and accurately solve Eq.~(\ref{eq:ouq_program_subint}), which has demonstrated superior performance in handling moment-constraint optimization problems in OUQ~\cite{stenger2020optimal2}.

Canonical moments can be interpreted as the relative positions of classical moments within their admissible moment spaces~\cite{dette1997theory}. They are defined on the interval \([0,1]\), and any finite sequence of classical moments can be transformed into an equivalent sequence of canonical moments of the same order. The key idea of the canonical-moment-based approach is to convert the constrained optimization problem, originally formulated in terms of the Dirac supports \(\{x\}\) and weights \(\{t\}\), into an unconstrained optimization problem parameterized by a set of canonical moments of higher order than those appearing in the constraints. These canonical moments, treated as optimization variables, are referred to as \emph{free canonical moments}. The moment constraints are enforced by constructing a polynomial function using \emph{fixed canonical moments}, which are derived from the given moment constraints. The roots of this polynomial correspond to the Dirac supports \(\{x\}\), and the associated weights \(\{t\}\) are subsequently determined by solving the linear system defined by the moment constraint equations. This canonical-moment-based formulation preserves the feasibility of the constraint space, ensuring nonnegativity of the Dirac weights, containment of supports within the prescribed subdomains, and exact satisfaction of the moment-matching conditions. 

For each subdomain \(\mathcal{X}_{i,j}\) and its corresponding truncated moments \(M_{i,j,k-1}\), the optimization problem is reformulated using canonical moments through the following steps: (1) transform the raw moments from the original subdomain \(\mathcal{X}_{i,j}\) to the unit interval \([0,1]\); (2) compute the first \(r\) fixed canonical moments from the given moment constraints and augment them with \(r+1\) free canonical moments that serve as optimization variables; (3) construct a tridiagonal Jacobi matrix whose eigenvalues correspond to the Dirac supports and whose eigenvectors determine the Dirac weights; (4) map the resulting discrete measure back to the original subdomain \(\mathcal{X}_{i,j}\); (5) formulate a constraint-free optimization problem in terms of these free canonical moments; and (6) reduce the number of indicator-function evaluations in the PoF calculation by ITS.

\subsubsection{Affine Transformation to the Unit Interval}
\label{sec:affine}

We first transform random inputs and the corresponding truncated raw moments from the original subdomain \(\mathcal{X}_{i,j}\) to the unit interval \([0,1]\). Specifically, each random input is affinely scaled using the min-max normalization
\begin{equation}
\xi_i = \frac{x_i - a_{i,j}}{b_{i,j} - a_{i,j}}, 
\quad x_i \in \mathcal{X}_{i,j},
\label{eq:affine_transform}
\end{equation}
where \(a_{i,j}\) and \(b_{i,j}\) denote the lower and upper bounds of \(\mathcal{X}_{i,j}\), respectively, and \(\xi_i \in [0,1]\) represents the normalized variable corresponding to the random input \(x_i\). The truncated raw moments are then transformed using a binomial expansion
\begin{equation}
\tilde{M}_{i,j,k-1}
= \frac{1}{(b_{i,j}-a_{i,j})^{k-1}}
\sum_{l=0}^{k-1} 
\binom{k-1}{l}
(-a_{i,j})^{k-l-1} M_{i,j,l},
\qquad k = 1, \dots, r+1.
\label{eq:moment_01}
\end{equation}

This affine transformation serves two primary purposes. First, it normalizes the total probability mass over the subdomain \(\mathcal{X}_{i,j}\), ensuring that \(\tilde{M}_{i,j,0} = 1\). Second, it standardizes the supports to the unit interval \([0,1]\), on which the theory of canonical moments is defined. By construction, each \(\tilde{M}_{i,j,k-1}\) represents the \((k-1)\)-th raw moment of a probability measure on \([0,1]\) corresponding to the original truncated measure on \(\mathcal{X}_{i,j}\).

\subsubsection{Computation of Canonical Moments}
\label{sec:canonical}

Given a sequence of \(r+1\) raw moment constraints \(\tilde{M}_{i,j,k-1}\) (\(k = 1, \dots, r+1\)) defined over the subdomain \(\mathcal{X}_{i,j}\), the corresponding fixed canonical moments \(p_{i,j,k-1}\) (\(k = 1, \dots, r+1\)) can be computed. The analytical expressions for the first three moments (\(k = 1, 2, 3\)) are available
\begin{equation}
\label{eq:p1_p2_explicit}
p_{i,j,0} = 0, 
\quad 
p_{i,j,1} = \tilde{M}_{i,j,1}, 
\quad 
p_{i,j,2} = 
\frac{\tilde{M}_{i,j,2} - (\tilde{M}_{i,j,1})^2}
{\tilde{M}_{i,j,1}(1 - \tilde{M}_{i,j,1})}.
\end{equation}

For higher-order moments, the corresponding canonical moments can be computed using two classical algorithms. The first is based on Hankel determinants, in which a sequence of square Hankel matrices is constructed from the raw moments, and the canonical moments are derived from the determinants of these matrices~\cite{dette1997theory}. However, this approach requires evaluating many Hankel determinants and is therefore computationally expensive and potentially numerically unstable. To overcome this issue, a more efficient alternative is the Quotient-Difference (Q-D) algorithm~\cite{henrici1993applied,gautschi2004orthogonal}, which recursively generates a triangular array---known as the Q-D table~\cite{dette1997theory}---and yields the canonical moments with significantly improved numerical efficiency. It is important to note that enforcing the original moment constraints \(M_{i,j,k-1}\) \((k = 1, \dots, r+1)\) or their canonical counterparts \(p_{i,j,k-1}\) \((k = 1, \dots, r+1)\) is mathematically equivalent in terms of moment matching.

Once the first \(r+1\) fixed canonical moments \(p_{i,j,k-1}\) have been determined, they are augmented by an additional set of \(r+1\) higher-order canonical moments, \(p_{i,j,k-1}\) \((k = r+2, \dots, 2r+2)\), which serve as optimization variables in the OUQ problem. These free canonical moments take values in the interval \((0,1)\), which ensures that the resulting measure remains admissible within the moment space.

\subsubsection{Solution of Supports and Weights via Eigendecomposition}
\label{sec:eigen}

Given the canonical moment sequence \(p_{i,j,k-1}\) \((k = 1, \dots, 2r+2)\) associated with the subdomain \(\mathcal{X}_{i,j}\), we seek to construct an \((r+1)\)-point Dirac measure \(\tilde{\mu}_{i,j} = \sum_{k=1}^{r+1} \omega_{i,j,k}\,\delta_{\xi_{i,j,k}} \) supported on \([0,1]\), whose first \(r+1\) moments match the transformed moment constraints \(\tilde{M}_{i,j,k-1}\) \((k = 1, \dots, r+1)\). Here, \(\xi_{i,j,k}\) and \(\omega_{i,j,k}\) denote the Dirac supports and weights, respectively. Following~\cite{stenger2020optimal}, the supports \(\xi_{i,j,k}\) \((k = 1, \dots, r+1)\) coincide with the roots of a polynomial \(P_{(r+1)}\) of degree \(r+1\), which is generated through the three-term recurrence relation
\begin{equation}
\label{eq:three_term_recurrence}
P_{(k+1)}(\xi)
=
\bigl(\xi - (\zeta_{2k+1} + \zeta_{2k+2})\bigr) P_{(k)}(\xi)
-
\zeta_{2k}\,\zeta_{2k+1}\,P_{(k-1)}(\xi),
\qquad k = 0,1,\dots,r,
\end{equation}
with the initial polynomials \(P_{(-1)}:= 0\) and \(P_{(0)}:= 1\). The recurrence coefficients \(\zeta_k\) over \(\mathcal{X}_{i,j}\) are defined as
\begin{equation}
\zeta_k := p_{i,j,k-1}\bigl(1 - p_{i,j,k}\bigr),
\qquad 
\zeta_0 := 0,\quad 
\zeta_{2r+2} := 0.
\label{eq:recurrence}
\end{equation}

Once the supports \(\xi_{i,j,k}\) are computed as the roots of \(P_{(r+1)}\), the weights \(\omega_{i,j,k}\) are obtained by enforcing the moment-matching conditions
\begin{equation}
\label{eq:moment}
\sum_{k=1}^{r+1} \omega_{i,j,k}\,\xi_{i,j,k}^{\,l-1}
=
\tilde{M}_{i,j,l-1},
\qquad
l = 1, 2, \dots, r+1.
\end{equation}
Eq.~\eqref{eq:moment} forms a linear system in unknown weights \(\omega_{i,j,k}\), whose coefficient matrix is a Vandermonde matrix constructed from the supports \(\xi_{i,j,k}\). This system is invertible as long as all supports \(\xi_{i,j,k}\) are distinct. However, the polynomial-based computation of the supports can become unstable when two or more roots are very close together or coincide. In such cases, the Vandermonde matrix becomes ill-conditioned, leading to substantial numerical errors in the recovered weights.

Alternatively, in this work we employ an eigendecomposition-based procedure known as the Golub-Welsch algorithm~\cite{golub1969calculation}. A symmetric tridiagonal Jacobi matrix \(J_{i,j}\) over \(\mathcal{X}_{i,j}\) is constructed from the recurrence coefficients
\begin{equation}
\label{eq:jacobi_matrix_unit}
J_{i,j} =
\begin{bmatrix}
\zeta_1{+}\zeta_2 & \sqrt{\zeta_2\zeta_3} & & \\
\sqrt{\zeta_2\zeta_3} & \zeta_3{+}\zeta_4 & \ddots & \\
 & \ddots & \ddots & \sqrt{\zeta_{2r}\zeta_{2r+1}} \\
 & & \sqrt{\zeta_{2r}\zeta_{2r+1}} & \zeta_{2r+1}{+}\zeta_{2r+2}
\end{bmatrix}
\in \mathbb{R}^{(r+1)\times(r+1)} .
\end{equation}
The eigenvalues of the Jacobi matrix \(J_{i,j}\) equal to the Dirac supports \(\xi_{i,j,k}\), while the corresponding weights are obtained from the normalized eigenvectors
\begin{equation}
\omega_{i,j,k} = v_{k,1}^{\,2},
\label{eq:weight}
\end{equation}
where \(v_{k,1}\) is the first component of the \(k\)-th eigenvector. 

Constructing the Jacobi matrix requires canonical moments up to order \(2r+2\): the first \(r+1\) fixed canonical moments, derived from the truncated moment constraints, and additional \(r+1\) free canonical moments of higher order. Any admissible choice of these free moments \(p_{i,j,k-1}\) \((k=r+2,\dots,2r+2)\), together with the fixed moments \(p_{i,j,k-1}\) \((k=1,\dots,r+1)\), uniquely determines a probability measure satisfying the moment-matching conditions in Eq.~(\ref{eq:moment}). Thus, these higher-order canonical moments serve as the optimization variables in the resulting constraint-free OUQ formulation. Because canonical moments are naturally bounded within \([0,1]\), admissible candidate solutions are straightforward to generate, in contrast to the classical-moment representation.

From a numerical standpoint, the eigendecomposition-based procedure replaces the polynomial root-finding step with the eigendecomposition of a real, symmetric, tridiagonal matrix---a well-conditioned problem for which fast and stable algorithms are readily available. In addition, the corresponding eigenvectors directly provide the Dirac weights, thereby eliminating the need to solve a potentially ill-conditioned Vandermonde system.

\subsubsection{Recovery of Supports and Weights}
\label{sec:recovery}

Once the Dirac supports \(\xi_{i,j,k}\) and weights \(\omega_{i,j,k}\) on the unit interval \([0,1]\) are obtained via the eigendecomposition of the Jacobi matrix, they are mapped back to the original subdomain \(\mathcal{X}_{i,j}\) using
\begin{equation}
\label{eq:atom_weight_recovery}
x_{i,j,k} = a_{i,j} + (b_{i,j}-a_{i,j})\,\xi_{i,j,k}, 
\qquad
t_{i,j,k} = M_{i,j,0}\,\omega_{i,j,k},
\qquad
i=1,\dots,m,\; j=1,\dots,K,\; k=1,\dots,r+1.
\end{equation}
The recovered supports \(x_{i,j,k}\) and weights \(t_{i,j,k}\) on the original subdomain are then supplied to the forward model \(F\) for the evaluation of the PoF.

\subsubsection{Constraint-Free Optimization over Canonical Moments}

After transforming the classical moments into canonical moments and introducing additional higher-order canonical moments, the constrained OUQ problem in Eq.~\eqref{eq:ouq_program_subint} can be reformulated as an unconstrained optimization problem over a hypercube with the following objective function
\begin{equation} 
\label{eq:ouq_canonical}
\begin{aligned}
    \text{PoF}\big(\{p\}^{\mathrm{free}}\big) = 
    \sum_{\alpha_1=1}^{K(r+1)} \!\!\cdots\!
    \sum_{\alpha_m=1}^{K(r+1)}
    \left(\prod_{i=1}^m t_{i,\alpha_i}\big(\{p\}\big)\right)
    \mathbf{1}
    \bigg[
        G\Big(
            x_{1,\alpha_1}\big(\{p\}\big),
            \dots,
            x_{m,\alpha_m}\big(\{p\}\big)
        \Big)
        \in \mathcal{Y}^\text{c}
    \bigg],
\end{aligned}
\end{equation}
where  
\[
\{p\}^{\mathrm{free}} \equiv 
\{p_{i,j,k-1} : i=1,\dots,m;\; j=1,\dots,K;\; k=r+2,\dots,2r+2\}
\in (0,1)^{\,mk(r+1)}
\]
denotes the free higher-order canonical moments, and  
\[
\{p\} \equiv 
\{p_{i,j,k-1} : i=1,\dots,m;\; j=1,\dots,K;\; k=1,\dots,2r+2\}
\]
contains both fixed and free canonical moments. The mappings from canonical moments to Dirac supports \(x_{i,\alpha_i}(\{p\})\) and weights \(t_{i,\alpha_i}(\{p\})\) are performed following the procedures in Sections~\ref{sec:affine}-\ref{sec:recovery}. Accordingly, the canonical-moment-based unconstrained formulation developed here is specific to the independent-input setting, in which each marginal measure can be parameterized separately on its partitioned subdomains. For dependent inputs, a different parameterization of the admissible set would be required.

It is important to note that this canonical-moment reformulation does not alter the discrete nature of the objective function, which still involves indicator functions. However, it dramatically reduces the number of optimization variables. Only \(mK(r+1)\) free canonical moments are optimized, compared with the original formulation requiring \(mK(r+1)\) Dirac supports and \(mK(r+1)\) associated weights---a total of \(2mK(r+1)\) variables. As a result, this approach not only converts the problem into a fully constraint-free optimization but also substantially decreases the dimensionality of the search space, leading to a more efficient and numerically stable optimization procedure.

\subsubsection{Model Reduction via Inverse Transform Sampling}
\label{sec:mc_sampling}

The objective function in the constraint-free OUQ formulation given in Eq.~\eqref{eq:ouq_canonical} requires \((K(r+1))^m\) evaluations of the forward map and its associated indicator function to compute a single PoF value. Consequently, the computational cost can become prohibitive, especially for high-dimensional problems. For example, in our numerical study with \(K=8\), \(r=2\), and \(m=5\), computing a single PoF evaluation would require nearly eight million forward-map and indicator-function evaluations.

To mitigate this cost, we approximate the PoF through ITS~\cite{robert1999monte, rubinstein2016simulation}, a classical method for generating random samples from an arbitrary probability distribution using its cumulative distribution function (CDF). For the discrete Dirac measure in Eq.~\eqref{eq:classical_measure_decomposition}, the CDF of the \(i\)-th random input is
\begin{equation}
    F_i(x)
    = \sum_{j=1}^{K} \sum_{k=1}^{r+1}
      t_{i,j,k}\,
      \mathbf{1}\!\left[x_{i,j,k} \le x\right].
\end{equation}

We draw samples via the ITS method by generating i.i.d. random numbers from a uniform distribution on \([0,1]\), i.e., \(c_l \sim \mathrm{Unif}[0,1]\) for \(l = 1,2,\dots, N_{\mathrm{ITS}}\), where \(N_{\mathrm{ITS}}\) denotes the number of inverse-transform samples. The \(l\)-th sample of the \(i\)-th input is then obtained by
\begin{equation}
\label{eq:inverse_transform}
    X_i^{(l)}
    =
    \inf\big\{x : F_i(x) \ge c_l \big\}, \quad i=1,2,\dots,m,
\end{equation}
which corresponds to selecting the smallest Dirac support whose cumulative weight exceeds the uniform draw. Fig.~\ref{fig:inverse} schematically illustrates the ITS procedure using a three-point Dirac measure as an example. With these sampled inputs, the PoF is approximated through Monte Carlo estimation
\begin{equation}
\label{eq:mc_estimator}
\text{PoF}\big(\{p\}^{\mathrm{free}}\big)
\;\approx\;
\frac{1}{N_{\mathrm{ITS}}}
\sum_{l=1}^{N_{\mathrm{ITS}}}
\mathbf{1}\!\left[
    G\big(X_1^{(l)},\dots,X_m^{(l)}\big)
    \in \mathcal{Y}^{\mathrm{c}}
\right].
\end{equation}

\begin{figure}[pos=htbp]
\centering
\includegraphics[width=0.6\textwidth]{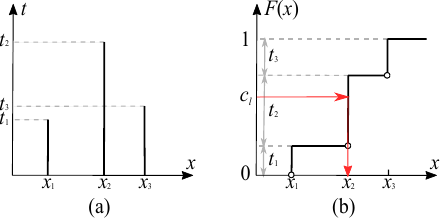}
\caption{Illustration of ITS in OUQ: (a) PDF of a three-point Dirac measure, and (b) Corresponding CDF and ITS.}
\label{fig:inverse}
\end{figure}

An additional motivation for using ITS arises from examining how each term of the joint Dirac measure contributes to the PoF. As shown in Eq.~\eqref{eq:ouq_canonical}, the Dirac supports determine whether a term contributes, while the Dirac weights determine the magnitude of that contribution. The contribution of a joint term is on the order of \(O(t^m)\), where \(t\) denotes a representative weight. Thus, for large \(m\), contributions from terms with small weights become negligible. In other words, computational effort should focus on combinations with relatively large weights. As illustrated in Fig.~\ref{fig:inverse}(b), ITS automatically emphasizes these dominant terms (e.g., $t_1$ compared with $t_2$ in Fig.~\ref{fig:inverse}(b)). Dirac supports with larger weights are selected more frequently, while those with very small weights are rarely sampled.

To compute a single PoF value, the ITS approach requires $N_{\mathrm{ITS}}$ evaluations of the forward map and its associated indicator function. The accuracy of ITS also depends on the number of samples: larger sample sizes generally yield more accurate approximations. Therefore, in practice, $N_{\mathrm{ITS}}$ should be chosen so that the sampling error is small relative to the changes in the PoF values being compared during optimization. However, for low-dimensional problems, the cost of ITS may exceed that of directly evaluating the PoF using all Dirac combinations, since \((K(r+1))^m\) may be smaller than \(N_{\mathrm{ITS}}\). Therefore, we introduce a threshold to decide whether ITS should be used. Specifically, if \((K(r+1))^m < N_{\mathrm{ITS}}\), we compute the PoF exactly using the original Dirac measure. Otherwise, when \((K(r+1))^m \ge N_{\mathrm{ITS}}\), we employ ITS with \(N_{\mathrm{ITS}}\) samples to approximate the PoF. 

Moreover, to stabilize the search for extremal PoF values, we employ the Common Random Numbers (CRN) variance reduction technique~\cite{carlo2001monte, rossetti2023simulation}. A single large vector of uniform random numbers is generated once and reused for the sampling step of every candidate measure evaluated. This induces positive correlation between the PoF estimates of similar candidates,  reducing the variance of their differences and enabling a more reliable and efficient optimization process.

\subsection{Implementation Details}

The OUQ methodology computes the extremal PoF values by optimizing over combinations of Dirac measures subject to truncated moment constraints. Algorithm~\ref{alg:cm_ouq} outlines the procedure for evaluating the objective function, namely, the PoF, in the resulting unconstrained canonical-moment formulation. It is important to emphasize that the inputs to the optimization algorithm are the \(mK(r+1)\) free canonical moments associated with all input variables, each taking values in the interval \((0,1)\).

In this work, we employ DE~\cite{storn1997differential} to search over the space of free canonical moments. DE is a population-based, derivative-free, stochastic optimization algorithm that is particularly well suited for OUQ. It handles the non-smooth, discontinuous nature of the indicator function in the PoF computation, naturally enforces simple box constraints, and performs reliably on global optimization problems with multiple local optima~\cite{price2005differential}. At each iteration, DE mutates and recombines a population of candidate solutions, decodes each candidate into Dirac supports and weights via Algorithm~\ref{alg:cm_ouq}, evaluates the corresponding PoF, and selects the best-performing candidates for the next generation. The procedure is repeated until convergence or until a prescribed maximum number of iterations is reached. In our implementation, the population size is chosen to be between $20$ and $50$ times the number of optimization variables, and the maximum number of iterations is set between $100$ and $200$, which was found to provide a good balance between convergence robustness and computational efficiency.

\begin{algorithm}[!ht]
\caption{Calculation of PoF as a function of free canonical moments}
\label{alg:cm_ouq}
\algblock{Begin}{End}
\begin{algorithmic}[1]

\State \textbf{Input}: Free canonical moments \(p_{i,j,k+r+1}\); subdomains \(\mathcal{X}_{i,j}\); truncated moments \(M_{i,j,k-1}\) (\(i=1,\dots,m\), \(j=1,\dots,K\), \(k=1,\dots,r+1\)); forward mapping \(G\); inadmissible set \(\mathcal{Y}^{\mathrm{c}}\); number of inverse-transform samples \(N_{\mathrm{ITS}}\)

\Begin

\For{$i = 1,\dots,m$}
\For{$j = 1,\dots,K$}
    \begin{enumerate}[label=(\alph*), leftmargin=2cm,  topsep=0pt, itemsep=1pt]
        \item Map the subdomain \(\mathcal{X}_{i,j}\) to \([0,1]\) and compute the transformed moments \(\tilde{M}_{i,j,k-1}\) using Eq.~\eqref{eq:moment_01}
        \item Compute fixed canonical moments \(p_{i,j,k}\) (\(k=1,\dots,r+1\)) using Eq.~\eqref{eq:p1_p2_explicit} or the Q-D algorithm
        \item Augment the fixed canonical moments with free canonical moments \(p_{i,j,k}\) (\(k=r+2,\dots,2r+2\)) to obtain the full sequence \(p_{i,j,k}\) (\(k=1,\dots,2r+2\))
        \item Compute recurrence coefficients via Eq.~\eqref{eq:recurrence} and construct the Jacobi matrix \(J_{i,j}\) using Eq.~\eqref{eq:jacobi_matrix_unit}
        \item Eigendecompose \(J_{i,j}\) to obtain Dirac supports \(\xi_{i,j,k}\) and weights \(\omega_{i,j,k}\) using Eq.~\eqref{eq:weight}
        \item Recover original supports \(x_{i,j,k}\) and weights \(t_{i,j,k}\) using Eq.~\eqref{eq:atom_weight_recovery}
    \end{enumerate}
\EndFor
\EndFor

\If{$(K(r+1))^m < N_{\mathrm{ITS}}$}
    \State Compute PoF exactly using Eq.~\eqref{eq:ouq_canonical}
\Else
    \For{$i = 1,\dots,m$}
        \State Generate \(N_{\mathrm{ITS}}\) uniform random samples on \([0,1]\)
        \State Generate \(N_{\mathrm{ITS}}\) samples of the \(i\)-th input using Eq.~\eqref{eq:inverse_transform}
    \EndFor
    \State Estimate PoF using Eq.~\eqref{eq:mc_estimator}
\EndIf

\End

\State \textbf{Output}: \(\mathrm{PoF}\big(\{p\}^{\mathrm{free}}\big)\)

\end{algorithmic}
\end{algorithm}

\section{A Special Case: OUQ under Zeroth-Order Moment Constraints}
\label{sec:evidence}

The OUQ framework under general moment constraints offers a powerful and flexible approach for modeling uncertainty and quantifying its impact on system outputs. It is natural to compare this framework with other established uncertainty models. To this end, in this section we examine a special case in which the uncertain inputs are constrained only by zeroth-order moments. We show that, under this setting, the resulting OUQ formulation becomes equivalent to the evidence theory~\cite{dempster2008upper, shafer2020mathematical}.

We assume that the zeroth moments, i.e., the probabilities of the input lying in each subdomain, are known. As a result, the information set $\mathcal{A}$ can be written as
\begin{equation}
    \mathcal{A} = 
    \bigg\{
    \mu ~\bigg|~
    \mu = \bigotimes_{i=1}^{m} \mu_i,~
    \mu_i\!\left[X_i \in \mathcal{X}_{i,j}\right] = M_{i,j,0},~
    i = 1, \dots, m;\;
    j = 1, \dots, K
    \bigg\}.
\label{eq:infset0}
\end{equation}

\begin{figure}[pos=htbp]
\centering
\includegraphics[width=0.4\textwidth]{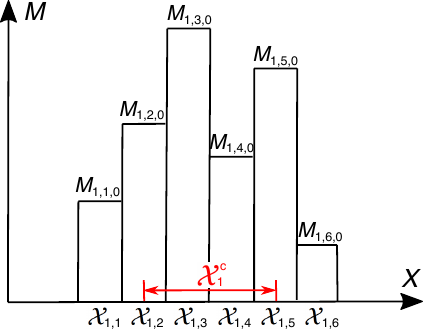}
\caption{Schematic illustration of the probability constraints for a one-dimensional random variable. The input domain is partitioned into subintervals $\mathcal{X}_{1,j}$, and the failure region $\mathcal{X}^\text{c}_1$ is highlighted for comparison. In this example, the upper and lower bounds on the PoF are $U=M_{1,2,0}+M_{1,3,0}+M_{1,4,0}+M_{1,5,0}$ and $L=M_{1,3,0}+M_{1,4,0}$, respectively.}
\label{fig:partition}
\end{figure}

A schematic illustration of partitioning the domain of a one-dimensional random input is shown in Fig.~\ref{fig:partition}. Our goal is to compute upper and lower bounds on $\mathbb{P}[Y \in \mathcal{Y}^{\mathrm{c}}]$ based on the available information in Eq.~\eqref{eq:infset0}. For any probability measure $\mu$, the PoF can be expressed as
\begin{equation}
	\mu\!\left[Y \in \mathcal{Y}^{\mathrm{c}}\right]
	= \mu\!\left[X \in \mathcal{X}^{\mathrm{c}}\right],
	\label{eq:muxy}
\end{equation}
where $\mathcal{X}^{\mathrm{c}}$ is the failure domain in the input space
\begin{equation}
	\mathcal{X}^{\mathrm{c}}
    \equiv 
    \big\{\, X~|~G(X) \in \mathcal{Y}^{\mathrm{c}} \big\}
    \subseteq \mathcal{X}.
    \label{eq:failuredomain}
\end{equation}

To exploit the structure of the input domain, we partition $\mathcal{X}$ using Eq.~(\ref{eq:union}) as
\begin{equation}
    \mathcal{X}
    =
    \prod_{i=1}^{m}\bigcup_{\alpha=1}^{K}\mathcal{X}_{i,\alpha}.
\end{equation}
Expanding the product of unions yields
\begin{equation}
    \mathcal{X}
    =
    \bigcup_{\alpha_1=1}^{K}\cdots\bigcup_{\alpha_m=1}^{K}
    \mathcal{R}_{\alpha},
    \label{eq:domaindecomp}
\end{equation}
where each $\mathcal{R}_{\boldsymbol{\alpha}}$ is a rectangular partition element defined by
\begin{equation}
    \mathcal{R}_{\alpha}
    :=
    \prod_{i=1}^{m}\mathcal{X}_{i,\alpha_i},
\end{equation}
with multi-index $\alpha=(\alpha_1,\dots,\alpha_m)\in\{1,\dots,K\}^m$. The collection $\mathcal{R}_\alpha$ $(\alpha\in\{1,\dots,K\}^m)$ forms a partition of $\mathcal{X}$. Based on this partition, the upper and lower bounds on the PoF over the information set $\mathcal{A}$ follow as
\begin{subequations}
\label{eq:bounds0}
\begin{equation}
  	U(\mathcal{A})
    =
    \sum_{\alpha_1=1}^{K}\cdots\sum_{\alpha_m=1}^{K}
    \left(
    \prod_{i=1}^{m} M_{i,\alpha_i,0}
    \right)
    \mathbf{1}\!\left[
    \mathcal{X}^{\mathrm c}\cap\mathcal{R}_{\boldsymbol{\alpha}}\neq\varnothing
    \right],
\end{equation}    
and
\begin{equation}
    L(\mathcal{A})
    =
    \sum_{\alpha_1=1}^{K}\cdots\sum_{\alpha_m=1}^{K}
    \left(
    \prod_{i=1}^{m} M_{i,\alpha_i,0}
    \right)
    \mathbf{1}\!\left[
    \mathcal{X}^{\mathrm c}\cap\mathcal{R}_{\boldsymbol{\alpha}}=\mathcal{R}_{\alpha}
    \right].
\end{equation}
\end{subequations}

A detailed derivation of Eq.~\eqref{eq:bounds0} is provided in Appendix~\ref{sec:deri}. In words, the upper bound sums the probabilities of all subdomains $\mathcal{R}_{\alpha}$ that intersect the failure region $\mathcal{X}^{\mathrm{c}}$, whereas the lower bound sums the probabilities of those subdomains entirely contained in the failure region. These quantities correspond to the \emph{plausibility} and \emph{belief} functions in the evidence theory~\cite{dempster2008upper, shafer2020mathematical}. An example of these bounds is illustrated in Fig.~\ref{fig:partition}.

\section{Numerical Examples}
\label{sec:examples}

In this section, we present five numerical examples to evaluate the proposed subdomain-based OUQ methodology across progressively more challenging settings. Case 1 considers one-dimensional baseline problems with identity mappings and four basic distributions, for which the reference PoFs are available by quadrature and can therefore be used to verify the computed OUQ bounds. Case 2 examines a five-dimensional smooth nonlinear problem to assess the canonical-moment formulation and the use of ITS in a moderate-dimensional setting. Case 3 introduces a two-dimensional non-smooth four-branch benchmark, including both a standard-probability and a low-probability subcase, in order to test the applicability of the framework beyond smooth performance functions. Case 4 considers an eight-dimensional rare-event roof-truss problem to investigate the behavior of the method in a high-dimensional tail-probability regime, where active-dimension refinement is employed to control computational cost. Finally, Case 5 presents a ten-dimensional ballistic impact problem, in which a surrogate model is used to make large-scale OUQ computations tractable and additional studies are carried out under varying failure thresholds. Taken together, these cases assess the proposed framework across smooth and non-smooth mappings, low- and high-dimensional settings, and both moderate- and rare-event failure regimes.

In all examples, the forward mapping $G$ produces a single output with a maximum allowable threshold $Y^{\mathrm{c}}$, and the inadmissible set is therefore $\mathcal{Y}^{\mathrm{c}} = [\,Y^{\mathrm{c}}, +\infty)$. Although the true distributions of inputs and their corresponding PDFs are specified for each example, in the OUQ setting we assume that these PDFs are unknown. The truncated moments required by OUQ are computed using the true PDFs solely for the purpose of benchmarking. Moreover, we consider only equally partitioned input domains, though the methodology can be straightforwardly extended to accommodate non-equally partitioned domains.

\subsection{Case 1: One-Dimensional Identity Functions Problems}
\label{sec:1D}

We consider the case in which the mapping $G$ is an identity function defined over a bounded one-dimensional domain, i.e., $G(X) = X$ for $X \in [a,b]$. Four basic probability distributions are examined: the normal, uniform, Weibull, and bimodal normal-mixture distributions. Each distribution is further truncated to the input domain $[a,b]$ according to
\begin{equation}
    f^*(x) = \frac{f(x)}{F(b) - F(a)} \, \mathbf{1}\big[x \in [a,b]\big],
    \label{eq:trunPDF}
\end{equation}
where $f(x)$ and $F(x)$ denote the PDF and CDF of the original untruncated distribution. Table~\ref{tab:1-d_cases} summarizes the original PDF $f(x)$, the input domain $[a,b]$, the output threshold $Y^{\mathrm{c}}$, and the reference PoF. Fig.~\ref{fig:1d_true_pdf} illustrates the truncated PDFs and the corresponding failure domains for the four distributions. In this one-dimensional example, ITS is not employed. Instead, the PoF in OUQ is computed exactly using Eq.~\eqref{eq:ouq_canonical}.

\begin{table}[width=\linewidth,cols=5,pos=h]
\caption{Summary of the four one-dimensional problems.}
\label{tab:1-d_cases}
\begin{tabular*}{\tblwidth}{@{} L L L L L @{}}
\toprule
Distribution & $f(x)$ & $[a,b]$ & $Y^{\mathrm{c}}$ & reference PoF \\
\midrule
Normal                  & $\phi(x;0,1)$ & $[-5,5]$ & $0.7$ & $0.2420$ \\
Uniform                 & $0.1$     & $[-5,5]$ & $1.7$ & $0.3300$ \\
Weibull                &  $\dfrac{k}{\lambda}\!\left(\dfrac{x}{\lambda}\right)^{k-1}\!e^{-(x/\lambda)^k}$ with $k{=}1.5$ and $\lambda{=}2$ & $[0,10]$ & $3.0$ & $0.1593$ \\
Bimodal normal-mixture  &  $0.35\,\phi(x;-2,1)+0.65\,\phi(x;2,1)$ & $[-5,5]$ & $1.3$ & $0.4927$ \\
\bottomrule
\end{tabular*}
\vspace{2pt}
\parbox{\tblwidth}{%
  \footnotesize\normalfont\raggedright
  \emph{Notes:} (i) $\phi(x;m,s)=\frac{1}{\sqrt{2\pi}\,s}
  \exp\!\big(-\tfrac{(x-m)^2}{2s^2}\big)$. 
  (ii) The reference PoF is calculated by quadrature.%
}
\end{table}

\begin{figure}[pos=htbp]
\centering
\includegraphics[width=0.9\textwidth]{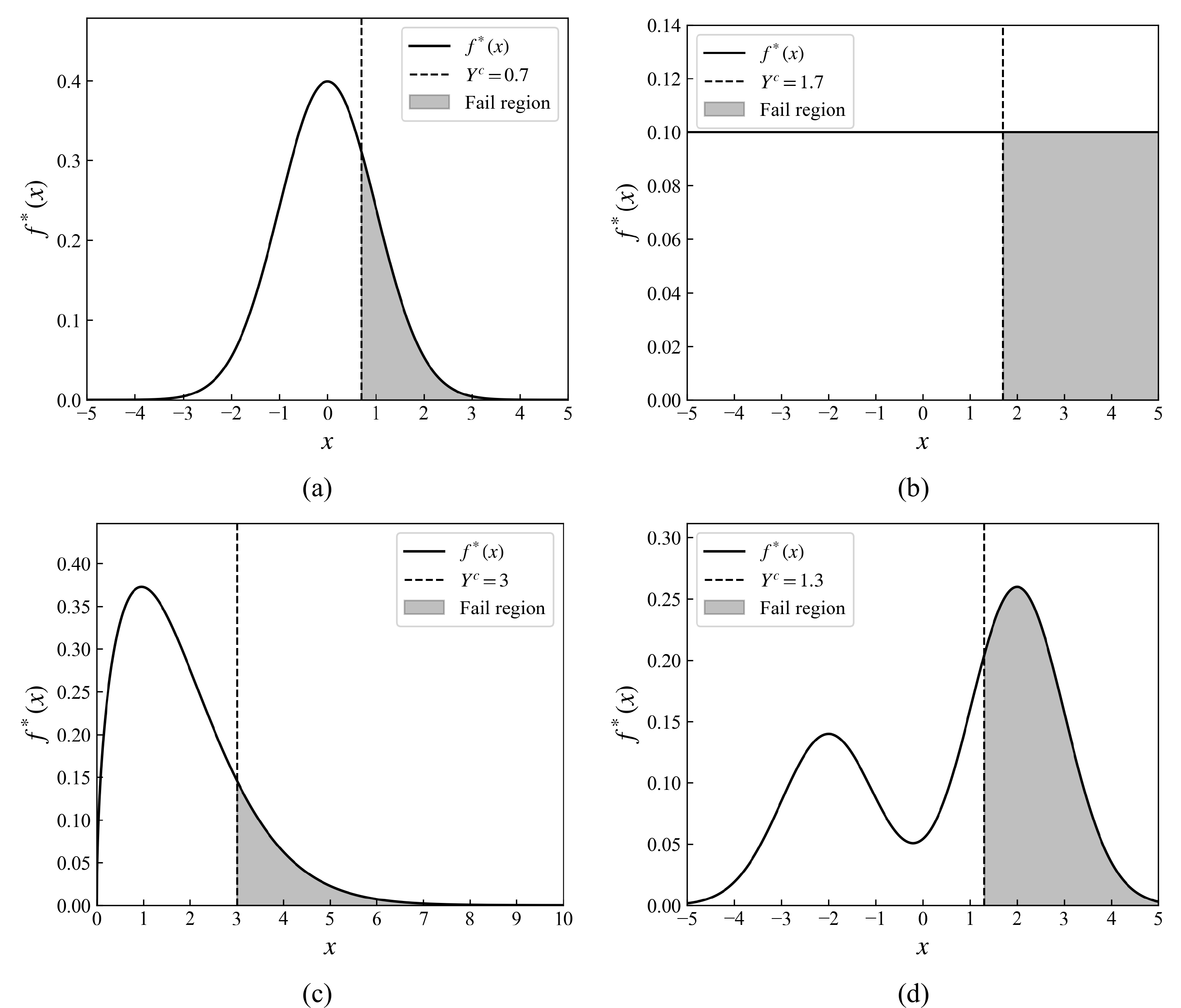}
\caption{Truncated PDFs and corresponding failure regions for the four considered distributions: (a) truncated normal, (b) uniform, (c) truncated Weibull, and (d) truncated bimodal normal-mixture.}
\label{fig:1d_true_pdf}
\end{figure}

\begin{figure}[pos=htbp]
\centering
\includegraphics[width=0.9\textwidth]{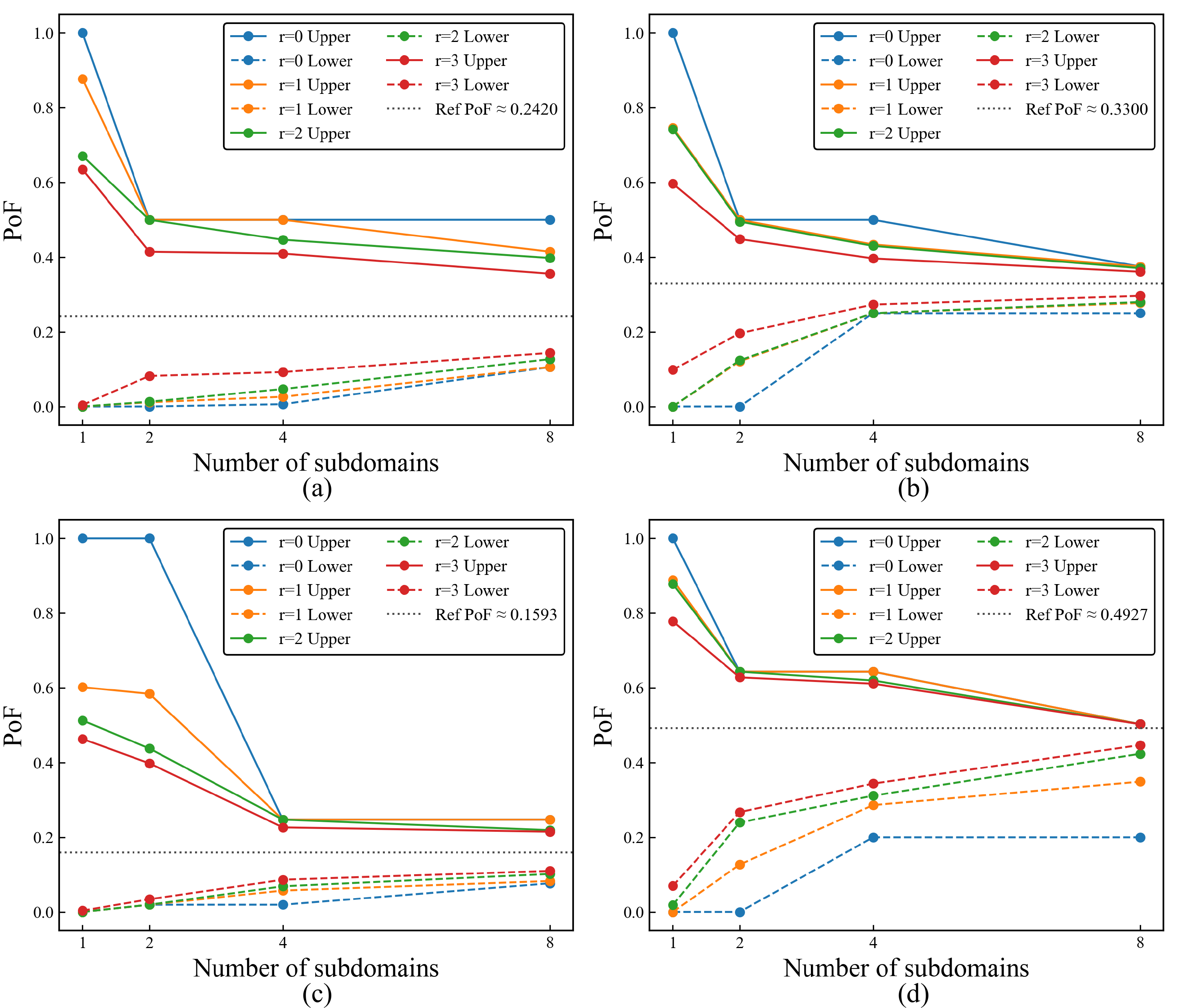}
\caption{Optimal PoF bounds as functions of the number of subdomains and moment constraints: (a) truncated normal, (b) uniform, (c) truncated Wei-bull, and (d) truncated bimodal normal-mixture.} 
\label{fig:1d_pof_bounds}
\end{figure}

The input domains of the four test cases are partitioned into $K = 1$, $2$, $4$, and $8$ subdomains. We consider moment constraints up to the third order, i.e., $r = 0$, $1$, $2$, and $3$. Recall that the case $r=0$ corresponds to the setting of the evidence theory. Fig.~\ref{fig:1d_pof_bounds} presents the reference PoF, along with its optimal upper and lower bounds obtained from the proposed OUQ framework. The numerical values of these optimal bounds are summarized in Tables~\ref{tab:ouq_1d_normal}-\ref{tab:ouq_1d_bimodal} in Appendix~\ref{sec:tables}.

A clear trend observed in Fig.~\ref{fig:1d_pof_bounds} is that the upper bound $U$ decreases and the lower bound $L$ increases as either the number of subdomains or the number of moment constraints increases. Consequently, the gap between $U$ and $L$ shrinks as additional uncertainty information is incorporated. Importantly, for all test cases, the reference PoF is consistently bracketed by the two optimal bounds. Since these bounds are optimal given the available information, further tightening requires incorporating additional knowledge about the uncertainties. As expected, in the limit of sufficiently many subdomains or higher-order moment constraints, both $U$ and $L$ should converge closely to the reference PoF. Thus, enriching the available uncertainty information enhances the tightness of the PoF bounds.

More specifically, the four distributions yield different optimal bounds for the same values of $K$ and $r$, except in certain extreme cases where $U = 1$ or $L = 0$, which correspond to scenarios with minimal uncertainty information and therefore uninformative bounds. The tightest bounds are obtained when the uncertainty information is richest, i.e., at $K = 8$ and $r = 3$. The resulting optimal bounds are $[L,\,U] = [0.1436,\,0.3556]$ for the normal distribution, $[0.2963,\,0.3608]$ for the uniform distribution, $[0.1098,\,0.2148]$ for the Weibull distribution, and $[0.4467,\,0.5027]$ for the bimodal normal-mixture distribution. These intervals deviate from their respective reference PoFs by $[-40.66\%,\, +46.94\%]$, $[-10.21\%,\, +9.33\%]$, $[-31.05\%,\, +34.87\%]$, and $[-9.34\%,\, +2.03\%]$. Therefore, among the four cases, the uniform and bimodal normal-mixture distributions show the greatest sensitivity to additional uncertainty information, the normal distribution shows the least sensitivity, and the Weibull distribution exhibits intermediate behavior.

An additional observation in Fig.~\ref{fig:1d_pof_bounds} is that increasing the number of subdomains or moment constraints does not always strictly decrease the upper bound or increase the lower bound. This behavior arises from the simplicity of the identity mapping considered in this example. When the additional uncertainty information pertains to regions of the input space that lie entirely within either the safe domain or the failure domain, it does not affect the PoF bounds, which are primarily governed by the probability mass near the boundary between the safe and failure domains, as illustrated in Fig.~\ref{fig:1d_true_pdf}.

\begin{figure}[pos=htbp]
\centering
\includegraphics[width=0.9\textwidth]{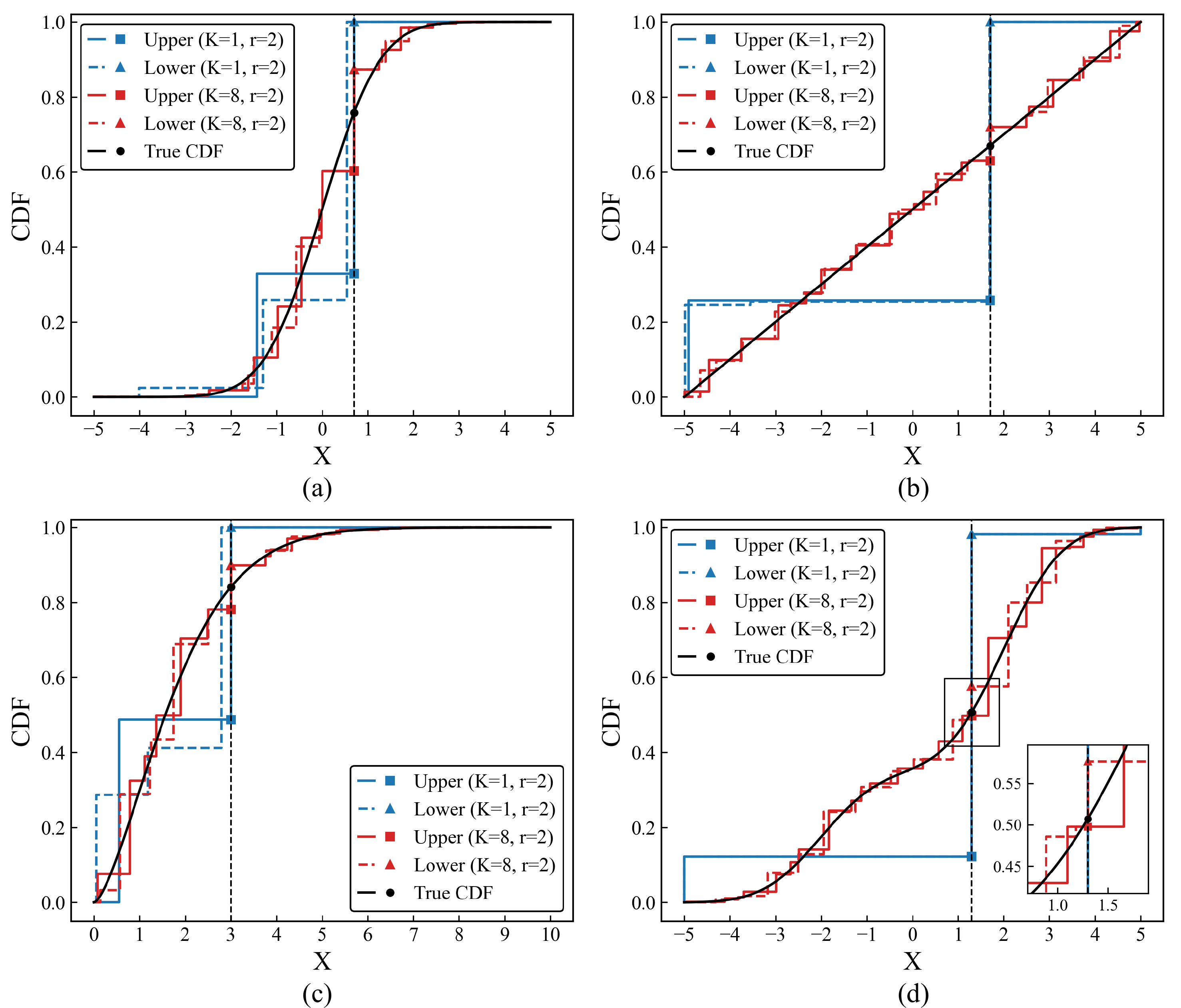}
\caption{Comparison between the Dirac CDFs at the optimal bounds and the corresponding true CDFs: (a) truncated normal, (b) uniform, (c) truncated Wei-bull, and (d) truncated bimodal normal mixture.}
\label{fig:1d_cdf_compare}
\end{figure}

To illustrate how uncertainty information influences the extremal measures, Fig.~\ref{fig:1d_cdf_compare} compares the Dirac CDFs obtained at the optimal bounds for the cases $K=1$ and $K=8$ (both with $r=2$) against the corresponding true probability CDFs. In these CDF plots, the supports of the Dirac measures appear as the locations of the step changes, while the height of each step represents the associated Dirac weight. As expected, as more information about the true measure is incorporated, the optimal-bound CDFs increasingly resemble the true CDF in all four cases. More uncertainty information leads to Dirac CDFs with more steps, yielding extremal measures that more closely approximate the true probability measures. It is also important to note, as shown in Fig.~\ref{fig:1d_cdf_compare}, that the obtained Dirac measures are extremal only with respect to the targeted PoF. They do \emph{not} provide optimal bounds on the CDF over its entire domain, in contrast to the behavior typically observed in p-box formulations~\cite{ferson2004arithmetic}.

\subsection{Case 2: A Five-Dimensional Nonlinear Smooth Problem}
\label{sec:5D}

We next consider a more complex five-dimensional nonlinear problem. The forward mapping $G$ is defined as
\begin{equation}
    G(X) = 0.7 X_1 + 0.35 X_2^2 - 0.25 X_3 X_4 + \frac{0.2 X_3^2}{1 + X_3^2} + \sin(X_5),
\end{equation}
where $X \equiv (X_1, \dots, X_5)$. The failure threshold for the output is $Y^{\mathrm{c}} = 0.8$. Each random input $X_i$ ($i=1,\dots,5$) is assumed to follow a truncated normal distribution computed using Eq.~(\ref{eq:trunPDF}). The input domains and the corresponding untruncated PDFs are summarized in Table~\ref{tab:5-d_case_input}.

\begin{table}[width=0.8\linewidth,cols=3,pos=h]
\caption{Summary of the inputs for the five-dimensional nonlinear problem.}
\label{tab:5-d_case_input}
\begin{tabular*}{\tblwidth}{@{} L L L @{}}
\toprule
Input & $f(x)$ & $[a,b]$ \\
\midrule
$X_1$ & $\phi(x;-1, 1.2)$ & $[-4,3]$  \\
$X_2$ & $\phi(x;0.5,0.7)$ & $[-2,4]$  \\
$X_3$ & $\phi(x;1.6,0.9)$ & $[-3.5,5]$  \\
$X_4$ & $\phi(x;-0.8,1.5)$ & $[-5,2]$  \\
$X_5$ & $\phi(x;0,0.6)$ & $[-2.5,2.5]$  \\
\bottomrule
\end{tabular*}
\vspace{2pt}
\parbox{\tblwidth}{%
  \footnotesize\normalfont\raggedright
  \emph{Note:} $\phi(x;m,s)=\frac{1}{\sqrt{2\pi}\,s}
  \exp\!\big(-\tfrac{(x-m)^2}{2s^2}\big)$. 
}
\end{table}

The input domains of the five variables are partitioned into $K = 1$, $2$, $4$, and $8$ subdomains simultaneously. We consider moment constraints up to the second order, i.e., $r = 0$, $1$, and $2$. Recall that a single exact evaluation of the PoF requires $(K(r+1))^{m}$ evaluations of the forward mapping. In this example, $m = 5$, so the exact computation becomes extremely expensive and even infeasible for large values of $K$ and $r$. Consequently, in these computationally prohibitive cases we employ the ITS method described in Section~\ref{sec:mc_sampling} to estimate the PoF within the OUQ framework. To determine an appropriate number of inverse-transform samples, we first examine how the PoF varies with the number of samples $N_{\mathrm{ITS}}$ for two representative cases: $K=4$ and $K=8$ both with $r=2$. For this comparison, we fix the free canonical moments $\{p\}^{\mathrm{free}}$ and hence the corresponding Dirac measures, and compute the resulting PoF.

\begin{figure}[pos=htbp]
\centering
\includegraphics[width=0.9\textwidth]{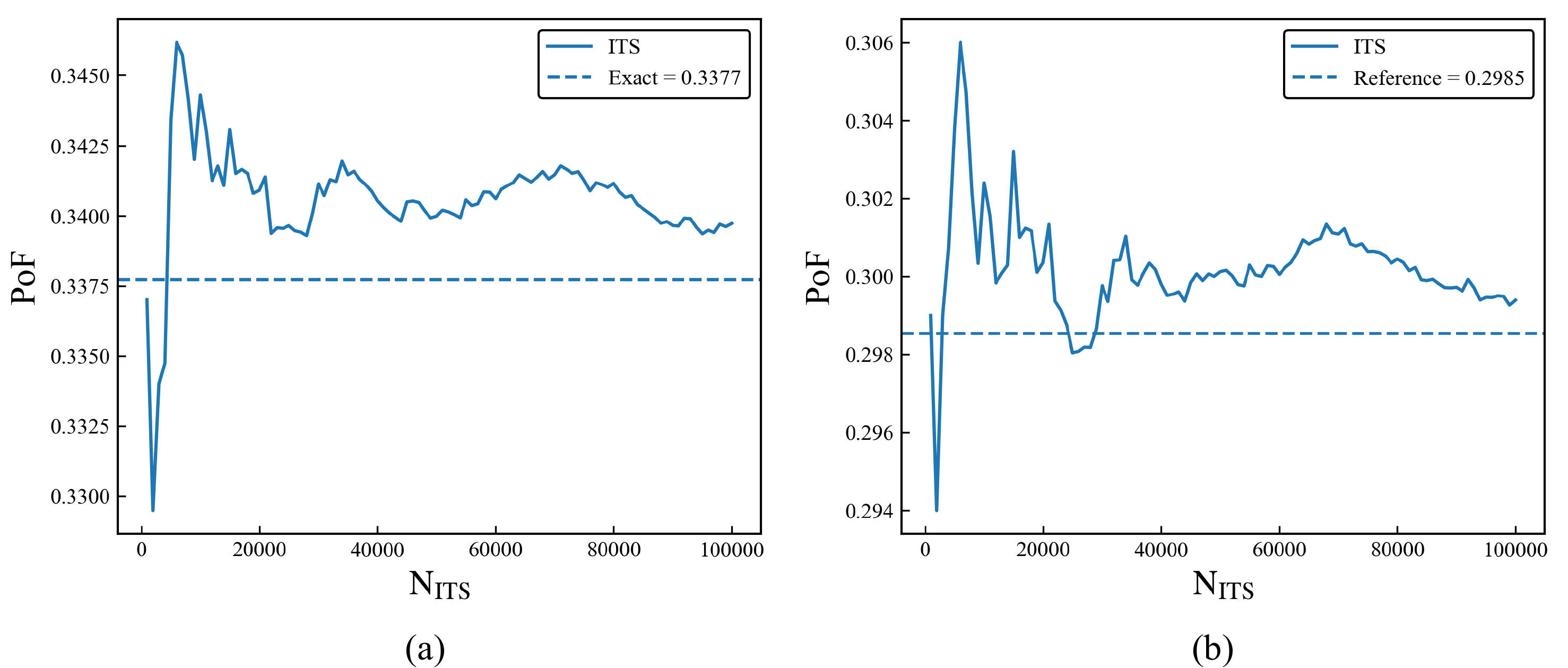}
\caption{Estimation of PoF using ITS: (a) $K=4$, $r=2$; and (b) $K=8$, $r=2$.}
\label{fig:5d_ouq_mc_exact_pof}
\end{figure}

Fig.~\ref{fig:5d_ouq_mc_exact_pof} shows the PoF as a function of $N_{\mathrm{ITS}}$ for these two cases. For the first case, Fig.~\ref{fig:5d_ouq_mc_exact_pof}(a), we also compute the exact PoF using Eq.~(\ref{eq:ouq_canonical}), which requires $248{,}832$ evaluations of the forward mapping. As expected, the estimate obtained via ITS converges to the exact value as $N_{\mathrm{ITS}}$ increases. In particular, when $N_{\mathrm{ITS}} > 4 \times 10^4$, the approximation error falls below $1\%$. For the second case, Fig.~\ref{fig:5d_ouq_mc_exact_pof}(b), computing the exact PoF would require $7{,}962{,}624$ forward evaluations, which is infeasible given our computational resources. Instead, we compute a high-fidelity reference value of the PoF using a large number of samples, specifically $N_{\mathrm{ITS}} = 2 \times 10^6$. The ITS estimate again exhibits convergence behavior. Once $N_{\mathrm{ITS}} > 2 \times 10^4$, the variation in the estimated PoF is less than $1\%$. Based on these observations, we choose $N_{\mathrm{ITS}} = 5 \times 10^4$ for all the computations in this and the following examples and apply the selection criterion described in Section~\ref{sec:mc_sampling} and Algorithm~\ref{alg:cm_ouq}. This choice balances accuracy and computational feasibility, reducing the computational cost by up to two orders of magnitude while maintaining acceptable accuracy. Fig.~\ref{fig:5d_ouq_mc_exact_pof}(b) also illustrates a practical guideline for selecting \(N_{\mathrm{ITS}}\): when ITS is used, a convergence study should be performed, and the final sample size should be chosen based on both convergence behavior and available computational resources.

\begin{figure}[pos=htbp]
\centering
\includegraphics[width=0.5\textwidth]{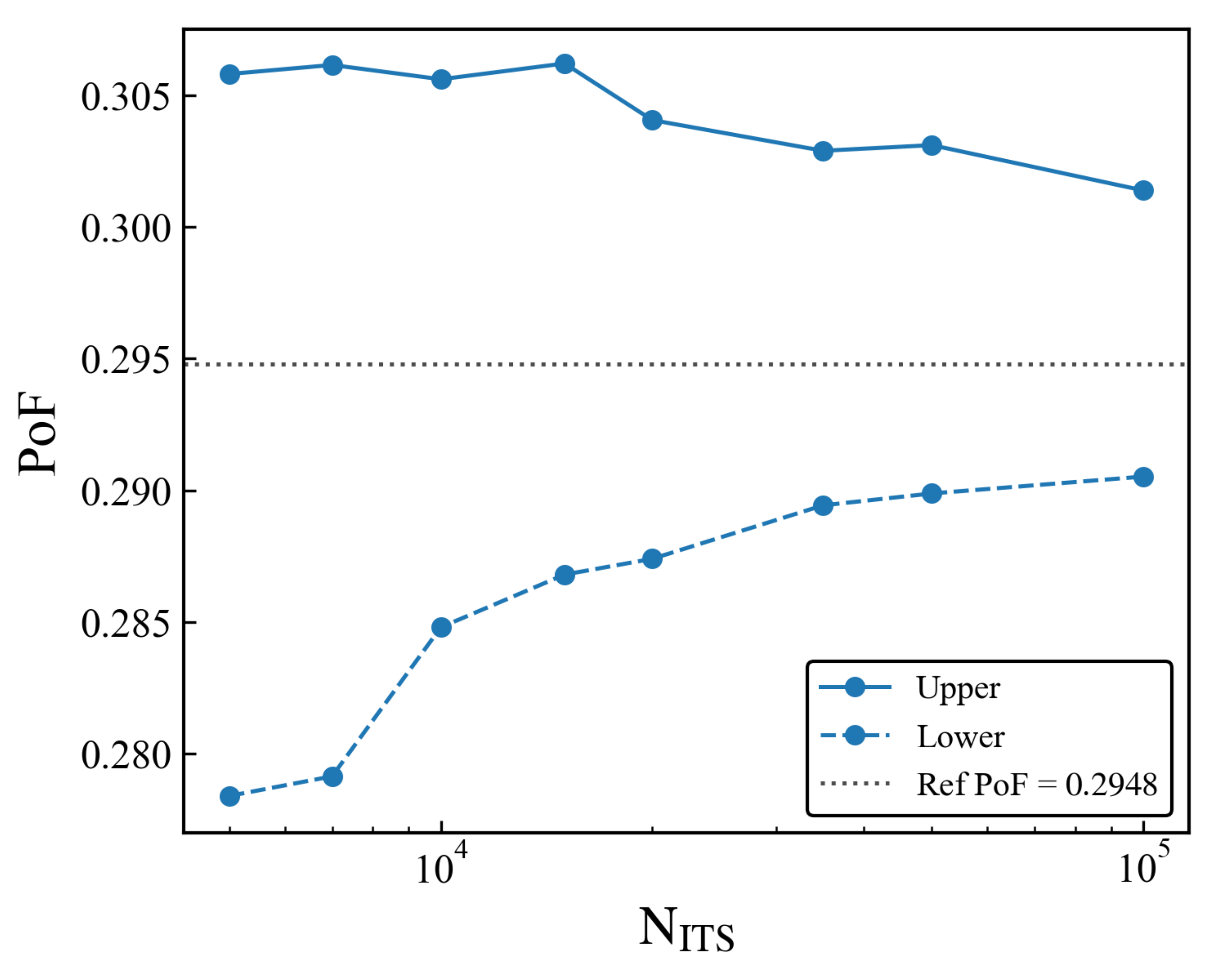}
\caption{Sensitivity of OUQ bounds to the ITS sample size $N_{\mathrm{ITS}}$ for the representative case $K=8$ and $r=2$.}
\label{fig:ouq_nits}
\end{figure}

We further examine the sensitivity of the OUQ bounds to the ITS sample size for the representative case $K=8$ and $r=2$, using $N_{\mathrm{ITS}} \in \{5\times10^3,\, 7\times10^3,\, 10^4,\, 1.5\times10^4,\, 2\times10^4,\, 3.5\times10^4,\, 5\times10^4,\, 10^5\}$. Fig.~\ref{fig:ouq_nits} reports the resulting upper and lower bounds as functions of $N_{\mathrm{ITS}}$, with the corresponding values listed in Table~\ref{tab:ouq_nits_scan} in Appendix~\ref{sec:tables}. As $N_{\mathrm{ITS}}$ increases, the upper bound decreases while the lower bound increases. The lower bound is more sensitive to $N_{\mathrm{ITS}}$. Both bounds stabilize by $N_{\mathrm{ITS}}=5\times10^4$. In particular, increasing $N_{\mathrm{ITS}}$ from $5\times10^4$ to $10^5$ changes the upper bound from $0.3031$ to $0.3014$ (a $0.56\%$ decrease) and the lower bound from $0.2899$ to $0.2905$ (a $0.21\%$ increase). In contrast, the computational cost approximately doubles, since it scales linearly with $N_{\mathrm{ITS}}$. These results indicate that, for this representative difficult case, sample sizes below $5\times10^4$ still under-resolve the extremal measures, whereas increasing $N_{\mathrm{ITS}}$ beyond $5\times10^4$ yields only marginal changes relative to the additional cost. Accordingly, we adopt $N_{\mathrm{ITS}}=5\times10^4$ as the default sample size.

\begin{figure}[pos=htbp]
\centering
\includegraphics[width=1\textwidth]{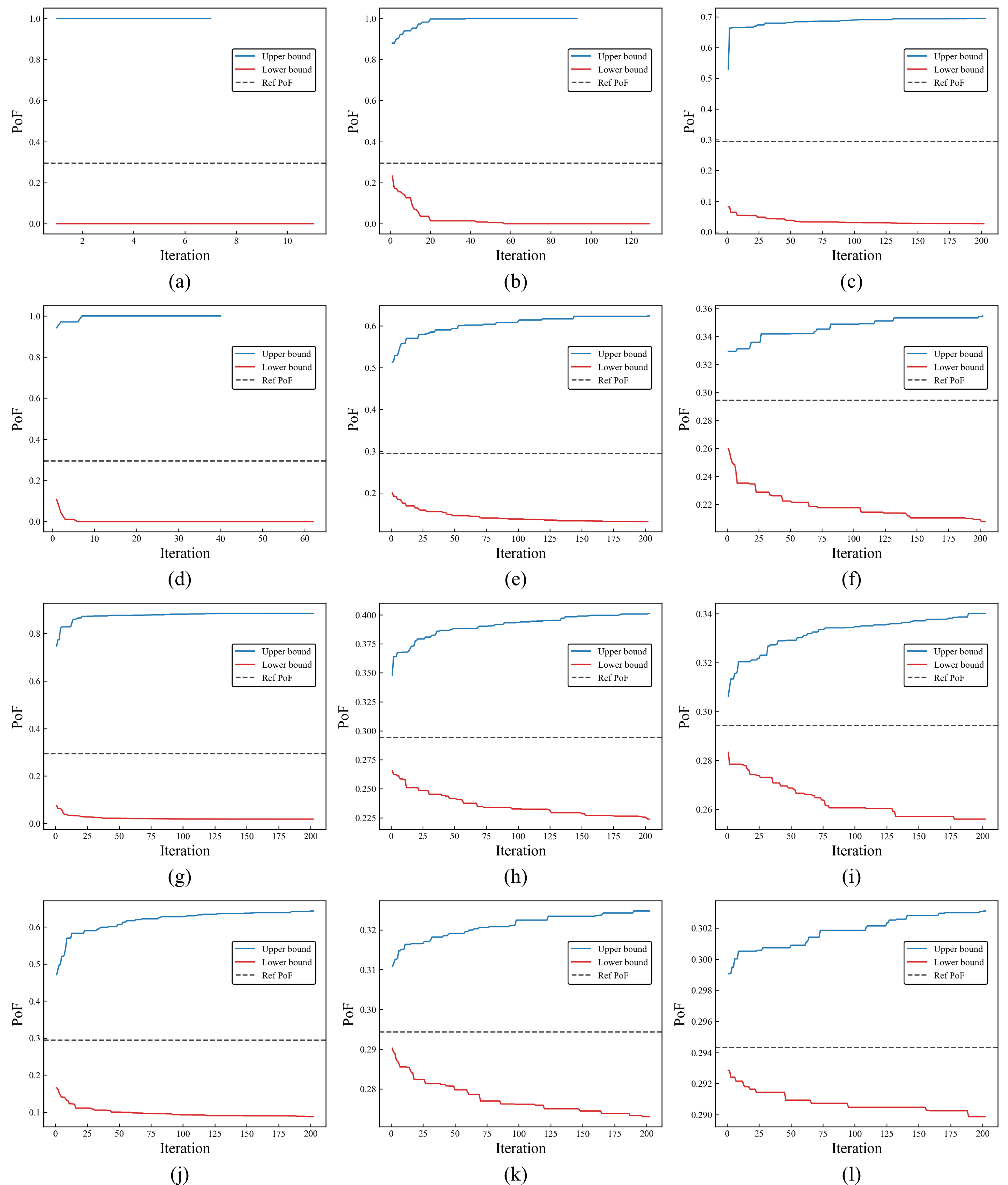}
\caption{DE convergence under different situations: (a) $K=1$, $r=0$; (b) $K=1$, $r=1$; (c) $K=1$, $r=2$; (d) $K=2$, $r=0$; (e) $K=2$, $r=1$; (f) $K=2$, $r=2$; (g) $K=4$, $r=0$; (h) $K=4$, $r=1$; (i) $K=4$, $r=2$; (j) $K=8$, $r=0$; (k) $K=8$, $r=1$; and (l) $K=8$, $r=2$.}
\label{fig:5d_ouq_bound_converge}
\end{figure}

Fig.~\ref{fig:5d_ouq_bound_converge} shows the best objective value at each generation during the DE iterations, together with the reference PoF estimated via Monte Carlo sampling using $2 \times 10^6$ samples and the underlying truncated normal distributions. Remarkably, despite the discreteness, high dimensionality, and strong nonlinearity of the OUQ problem, all DE runs converge to stable optima within $200$ generations. In particular, the objective values change by less than $1\%$ over the final $100$ iterations. This rapid convergence is largely attributable to the relatively large population size employed, set to $50$ times the number of optimization variables in this example. Across all situations, the reference PoF remains consistently bracketed by the optimal bounds throughout the iterations.

\begin{figure}[pos=htbp]
\centering
\includegraphics[width=0.5\textwidth]{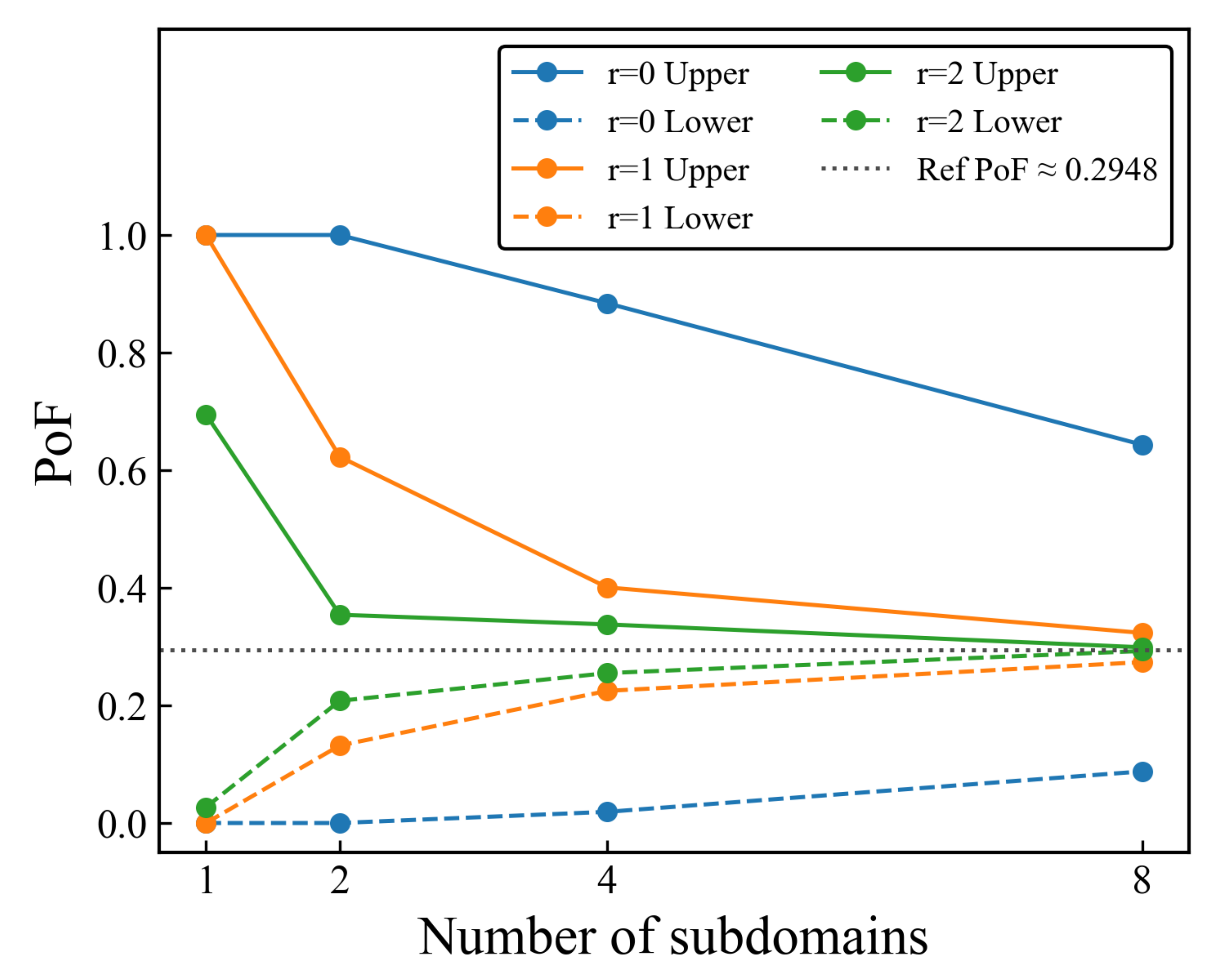}
\caption{Optimal PoF bounds for the five-dimensional nonlinear problem.}
\label{fig:5d_ouq_bound_compare}
\end{figure}

Fig.~\ref{fig:5d_ouq_bound_compare} compares the optimal bounds with the reference PoF. The numerical values of the optimal bounds are summarized in Table~\ref{tab:ouq_5d_bounds} in Appendix~\ref{sec:tables}. As expected, increasing either the number of subdomains or the number of moment constraints decreases the upper bound $U$ and increases the lower bound $L$, thereby tightening the bound interval. In the cases with the least uncertainty information---namely $K=1$, $r=0$; $K=1$, $r=1$; and $K=2$, $r=0$---the optimal bounds are $[L,\,U] = [0,\,1]$, which are uninformative. In contrast, in the case with the greatest amount of information, i.e., $K=8$, $r=2$, the optimal bounds are $[L,\,U] = [0.2925,~0.2988]$, deviating from the reference PoF $0.2948$ by only $[-0.61\%,~+1.53\%]$.

Another notable observation from Fig.~\ref{fig:5d_ouq_bound_compare} is that both the upper and lower bounds are more sensitive to the number of moment constraints than to the number of subdomains. This conclusion can be drawn by comparing the optimal bounds for the cases $K=4$, $r=0$ and $K=2$, $r=1$. Since the dimension of the OUQ optimization problem is $mK(r+1)$, both cases have the same computational cost in terms of evaluating the objective functions. However, the resulting bounds differ substantially. The bounds for $K=2$, $r=1$ are $[L,~U] = [0.1319,~0.6229]$, which are significantly tighter than those for $K=4$, $r=0$, namely $[L,~U] = [0.0190,~0.8839]$. A similar conclusion follows from comparing the cases $K=8$, $r=0$ and $K=4$, $r=1$. These findings suggest that, if the goal is to improve the tightness of the optimal bounds---and therefore the efficiency of certification---priority should be given to acquiring additional information in the form of higher-order local moment constraints rather than simply increasing the number of subdomains.

\subsection{Case 3: A Two-Dimensional Non-Smooth Four-Branch Problem}
\label{sec:2D}

We next consider a two-dimensional non-smooth benchmark problem based on the classical four-branch function~\cite{moustapha2024reliability}, which is widely used in structural reliability studies as a representative multi-branch failure system. In contrast to the smooth mappings considered in the preceding examples, the present benchmark is non-smooth because the system response is defined as the minimum of four branch functions, and hence exhibits branch switching along the interfaces where different branches become active. This feature makes the example well suited for examining the applicability of the proposed subdomain-based OUQ framework beyond smooth performance functions.

The four-branch problem is defined on a two-dimensional input vector $X \equiv (X_1,X_2)$, where $X_1$ and $X_2$ are independent truncated standard
normal random variables supported on $[-5,5]$. The system performance function is
\begin{equation}
G(X) = \min\{g_1(X),g_2(X),g_3(X),g_4(X)\},
\end{equation}
where the four branch functions are
\begin{equation}
\left\{\begin{array}{l}
g_1(X)=3+0.1\left(X_1-X_2\right)^2-\frac{1}{\sqrt{2}}\left(X_1+X_2\right),\\
g_2(X)=3+0.1\left(X_1-X_2\right)^2+\frac{1}{\sqrt{2}}\left(X_1+X_2\right),\\
g_3(X)=\left(X_1-X_2\right)+\frac{P}{\sqrt{2}},\\
g_4(X)=\left(X_2-X_1\right)+\frac{P}{\sqrt{2}},
\end{array}\right.
\label{eq:2d_nonsmooth_prob}
\end{equation}
with $P=6$ as the standard reference setting.

Two subcases are considered. In the first case, we set $Y^{\mathrm{c}}=0$, which corresponds to the standard four-branch failure condition. In the second case, we increase the critical threshold to $Y^{\mathrm{c}}=2$, thereby shrinking the failure domain while preserving the same nonsmooth, multi-branch structure. This second setting is introduced to examine the behavior of the proposed method in a low-probability (rare-event) regime without modifying the underlying benchmark function. The reference PoF for each subcase is estimated by direct Monte Carlo simulation using $2\times10^7$ samples. The resulting reference PoFs are approximately $4.51\times10^{-3}$ and $1.16\times10^{-5}$ for the first and second cases, respectively.

\begin{figure}[pos=htbp]
\centering
\includegraphics[width=0.9\textwidth]{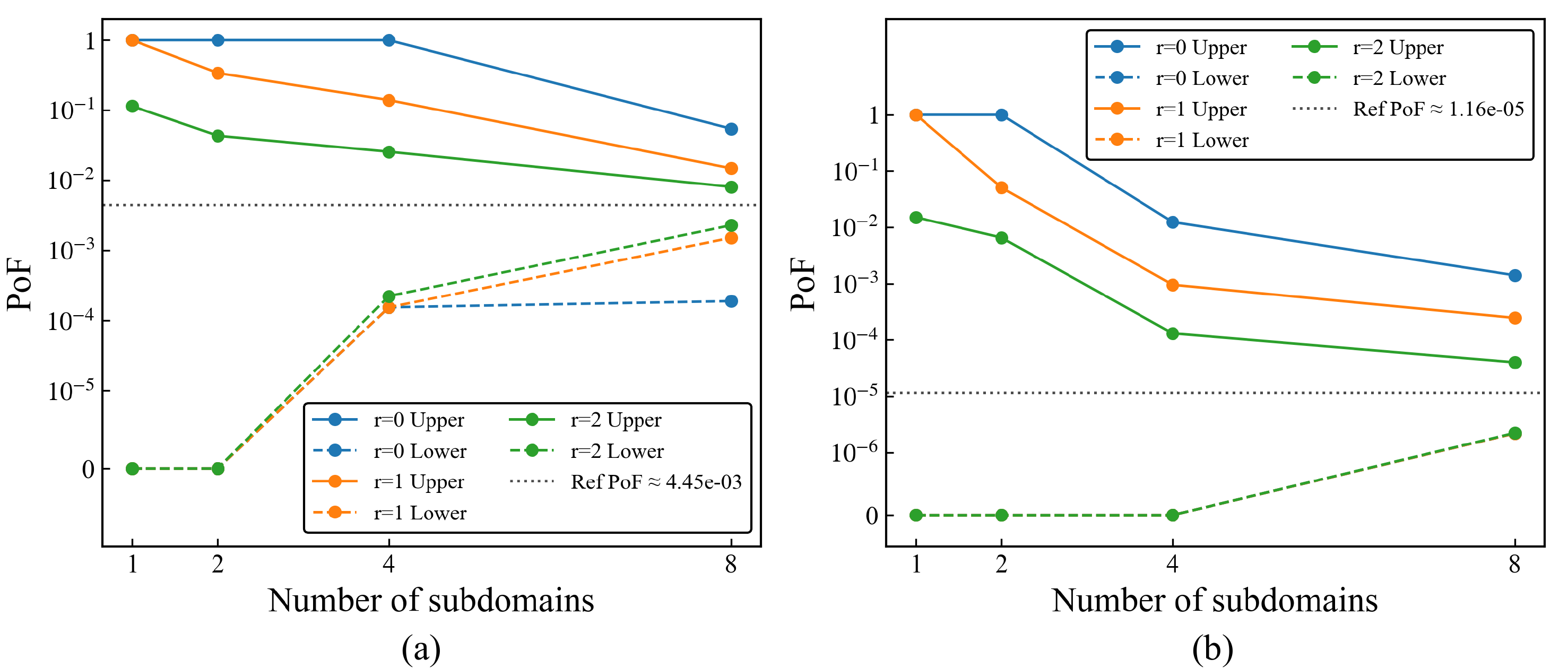}
\caption{Optimal PoF bounds for the two-dimensional non-smooth four-branch problem: (a) baseline case with $Y^{\mathrm{c}}=0$; and (b) low-probability case with $Y^{\mathrm{c}}=2$.}
\label{fig:2d_ouq_compare}
\end{figure}

We partition the domain $[-5,5]$ of both inputs, and Fig.~\ref{fig:2d_ouq_compare} compares the resulting optimal bounds with the reference PoF. The PoF values on the y-axis are plotted on a logarithmic scale. Detailed bound values are reported in Appendix Tables~\ref{tab:ouq_case3a_fourbranch} and~\ref{tab:ouq_case3b_fourbranch}. The results for both subcases show the same overall trend observed in the previous examples: increasing either the number of subdomains or the order of the moment constraints systematically tightens the PoF bounds. When only limited information is imposed, such as $K=1$ with $r=0$ or $r=1$, the bounds in both subcases remain nearly vacuous. As more local information is incorporated, the upper bound decreases substantially, by more than three orders of magnitude, while the lower bound increases.

For example, Fig.~\ref{fig:2d_ouq_compare}(a) shows that, for $Y^{\mathrm{c}}=0$, $K=8$, and $r=2$, the computed bounds are approximately $[L,\,U]=[2.29\times10^{-3},\,8.09\times10^{-3}]$, which bracket the reference PoF $4.51\times10^{-3}$ within a relatively narrow interval. For the low-probability rare-event case with $Y^{\mathrm{c}}=2$, Fig.~\ref{fig:2d_ouq_compare}(b) shows that the bounds at $K=8$ and $r=2$ are approximately $[L,\,U]=[2.23\times10^{-6},\,4.00\times10^{-5}]$, which still bracket the reference PoF $1.16\times10^{-5}$. With the richest uncertainty information considered here, both upper bounds are reduced to the same order of magnitude as their corresponding reference PoFs. These results confirm that the proposed subdomain OUQ framework can effectively handle non-smooth, multi-branch performance functions and remains feasible and informative for low-probability rare events in low-dimensional settings. A comparison between the two cases from $K=2$ to $K=4$ further suggests that the lower bound in the first case is more sensitive to the number of subdomains, whereas the upper bound in the second case is more sensitive to subdomain refinement.

\subsection{Case 4: An Eight-Dimensional Rare-Event Roof-Truss Problem}
\label{sec:8D}

To further examine the applicability of the proposed subdomain-based OUQ framework to high-dimensional rare-event problems, we consider an eight-dimensional roof-truss benchmark~\cite{moustapha2024reliability}, as shown in Fig.~\ref{fig:8d_roof_truss_schematic}. The structure consists of an assembly of concrete and steel members and is subjected to a distributed load on the top beams. This distributed loading is equivalently represented by point loads of magnitude $ql/4$ applied at nodes $D$ and $F$, and a point load of magnitude $ql/2$ applied at node $C$, where $q$ denotes the line load and $l$ is the span of the roof base. The uncertain input vector is \(X\equiv(q,l,A_s,A_c,E_s,E_c,f_s,f_c)\), whose physical meanings are summarized in Table~\ref{tab:8d_roof_truss_input} and illustrated in Fig.~\ref{fig:8d_roof_truss_schematic}. The roof-truss system exhibits three distinct failure modes. The first mode is associated with the vertical displacement at the roof apex C, which is required to remain below the critical value of $3$ cm. The second mode corresponds to compressive failure of member AD, whose internal force is $1.185\,ql$ and must not exceed its ultimate compressive resistance $f_cA_c$. The third mode corresponds to tensile failure of member EC, whose internal force is $0.75\,ql$ and must not exceed its ultimate tensile resistance $f_sA_s$. These three component limit states are written as
\begin{equation}
    \left\{\begin{array}{l}
    g_1(X)=0.03-\dfrac{q l^2}{2}\left(\dfrac{3.81}{A_cE_c}+\dfrac{1.13}{A_sE_s}\right), \\
    g_2(X)=f_cA_c-1.185\, ql, \\
    g_3(X)=f_sA_s-0.75\, ql.
    \end{array}\right.
\label{eq:8d_roof_truss_failure_modes}
\end{equation}
We introduce the scalar system output
\begin{equation}
    G(X)=-\min\{g_1(X),g_2(X),g_3(X)\},
\label{eq:8d_roof_truss_output}
\end{equation}
with a threshold \(Y^{\mathrm{c}}=1.25\times10^5\). This threshold differs from the failure criterion used in Ref.~\cite{moustapha2024reliability} and is introduced to create a rare-event setting.

Each input is modeled as an independent truncated lognormal marginal, with truncation bounds taken as the \(10^{-6}\) and \(1-10^{-6}\) quantiles of the corresponding lognormal distribution. The reference input distributions, means, coefficients of variation, and truncated lower and upper bounds are summarized in Table~\ref{tab:8d_roof_truss_input}. A system-level subset simulation~\cite{au2001estimation} is performed to estimate the reference PoF, yielding an extremely small value of approximately \(5.025\times10^{-7}\).

\begin{figure}[pos=htbp]
\centering
\includegraphics[width=0.7\textwidth]{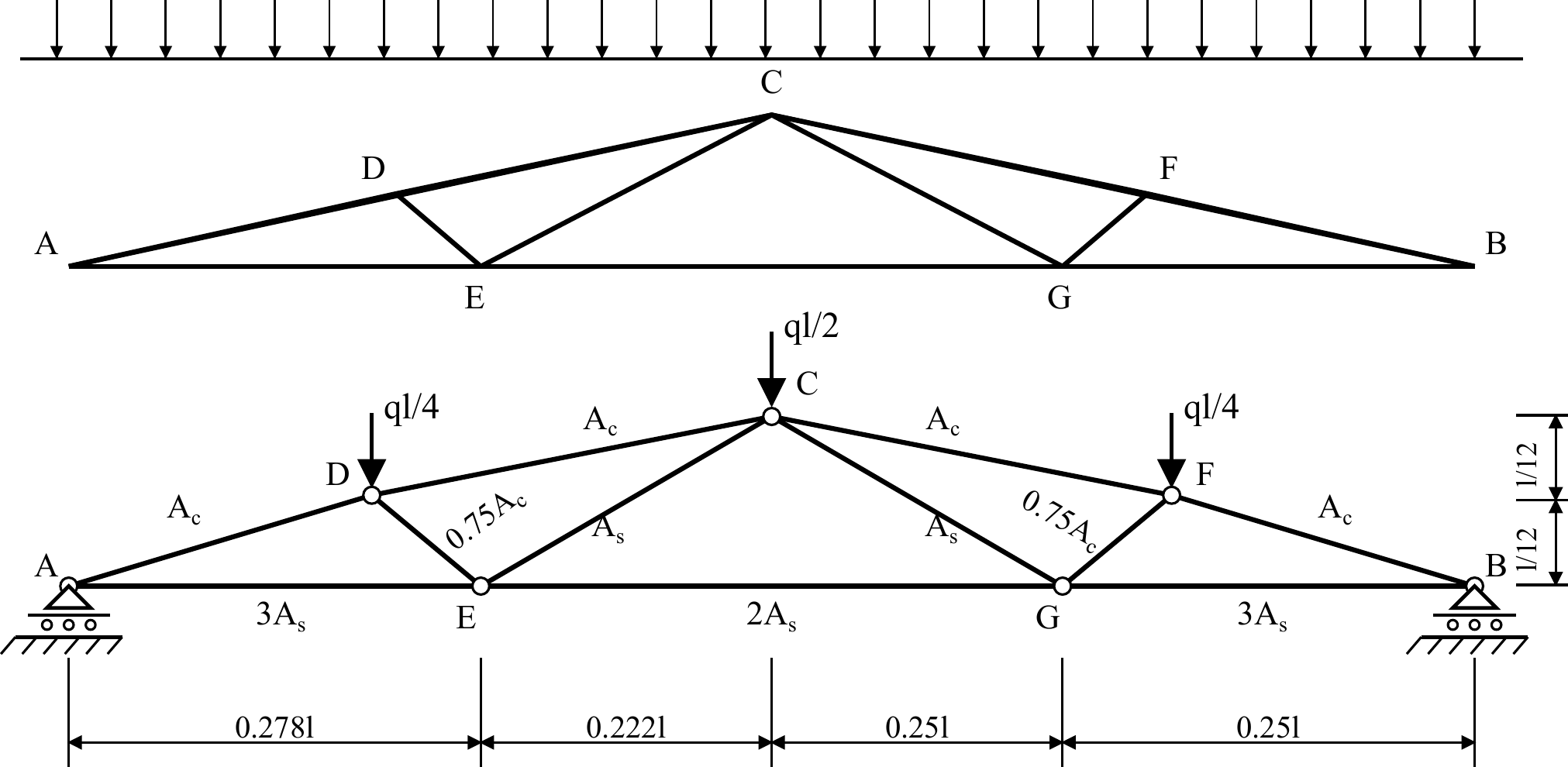}
\caption{Schematic illustration of the roof truss system~\cite{yun2019ak}.}
\label{fig:8d_roof_truss_schematic}
\end{figure}
\begin{table}[width=0.8\linewidth,cols=7,pos=htbp]
\caption{Probabilistic reference distribution of the input variables related to the roof truss problem~\cite{moustapha2024reliability}.}
\label{tab:8d_roof_truss_input}
\begin{tabular*}{\tblwidth}{@{} L L L L L L L@{}}
\toprule
Parameter & Unit & Distribution & Mean & CV & Lower bound & Upper bound \\
\midrule
Uniform load $q$  & [N/m] & Lognormal & $20,000$  & $0.07$ & $1.4310\times10^4$ & $2.7816\times10^4$ \\
Length $l$        & [m]   & Lognormal & $12$      & $0.01$ & $1.1442\times10^1$ & $1.2584\times10^1$ \\
Cross-sectional area $A_s$ & [m$^2$] & Lognormal & $9.82\times10^{-4}$  & $0.06$ & $7.3719\times10^{-4}$ & $1.3034\times10^{-3}$ \\
Cross-sectional area $A_c$ & [m$^2$] & Lognormal & $0.04$  & $0.12$ & $2.2496\times10^{-2}$ & $7.0113\times10^{-2}$\\
Young's modulus $E_s$ & [N/m$^2$] & Lognormal & $2\times10^{11}$  & $0.06$ & $1.5014\times10^{11}$ & $2.6546\times10^{11}$ \\
Young's modulus $E_c$ & [N/m$^2$] & Lognormal & $3\times10^{11}$  & $0.06$ & $2.2521\times10^{11}$ & $3.9819\times10^{11}$ \\
Tensile strength $f_s$ & [N/m$^2$] & Lognormal & $3.35\times10^{8}$  & $0.12$ & $1.8841\times10^8$ & $5.8719\times10^8$ \\
Compressive strength $f_c$ & [N/m$^2$] & Lognormal & $1.34\times10^{7}$  & $0.18$ & $5.6435\times10^6$ & $3.0819\times10^7$ \\
\bottomrule
\end{tabular*}
\end{table}

A central challenge in this case is that, when all eight input variables are refined simultaneously within the OUQ framework, exact PoF evaluation becomes computationally prohibitive. Moreover, in this extreme rare-event regime, simply increasing the number of inverse-transform samples is inefficient for PoF estimation. To address this issue, we preserve exact PoF evaluation on a reduced active set and control the overall computational cost through active-dimension refinement. Specifically, since \(q\) and \(l\) appear in all three limit-state functions in Eq.~\eqref{eq:8d_roof_truss_failure_modes}, they are retained in the reduced active set. A screening procedure based on McDiarmid subdiameters~\cite{sun2020rigorous} is then applied to the remaining six variables to prioritize additional inputs for domain refinement. For each input \(x_i\), the subdiameter is defined as
\begin{equation} 
\mathcal{D}_i=\sup_{\hat{x}_i\in \hat{I}_i,\;x_i,x_i'\in I_i} \left|G(\hat{x}_i,x_i)-G(\hat{x}_i,x_i')\right|, 
\end{equation}
where $I_i$ denotes the input domain of $x_i$, \(\hat{I}_i=\prod_{j\neq i} I_j\), \(\hat{x}_i=(x_1,\ldots,x_{i-1},x_{i+1},\ldots,x_m)\), and \((\hat{x}_i,x_i)\) denotes the full input vector formed by combining \(\hat{x}_i\) with the \(i\)-th coordinate value \(x_i\). In the present implementation, each \(\mathcal{D}_i\) is evaluated using DE-based global optimization over the remaining coordinates. The subdiameters are then normalized by their total sum, and the resulting ranking is used to identify the variables for which additional subdomain and moment information should be imposed.

The computed subdiameters are shown in Fig.~\ref{fig:8d_ouq_compare}(a), and the corresponding values are reported in Table~\ref{tab:McDiarmid_results}. The ranking indicates that \(f_c\) and \(A_c\) are the two most influential variables among the remaining six inputs, whereas \(E_s\) and \(E_c\) have negligible influence. Accordingly, in the OUQ analysis, only the four active variables \((q,l,A_c,f_c)\) are refined using subdomain partitioning and local moment constraints, while the remaining variables \((A_s,E_s,E_c,f_s)\) are kept at the coarsest resolution, namely \(K=1\) and \(r=0\). For the active set \((q,l,A_c,f_c)\), we consider \(K\in \{1,2,4,8\}\) subdomains and impose moment constraints up to second order, \(r\in \{0,1,2\}\). The resulting optimal PoF bounds are summarized in Fig.~\ref{fig:8d_ouq_compare}(b), and the detailed bound values are listed in Table~\ref{tab:ouq_case4_rooftruss_bounds} in Appendix~\ref{sec:tables}.

\begin{figure}[pos=htbp]
\centering
\includegraphics[width=0.9\textwidth]{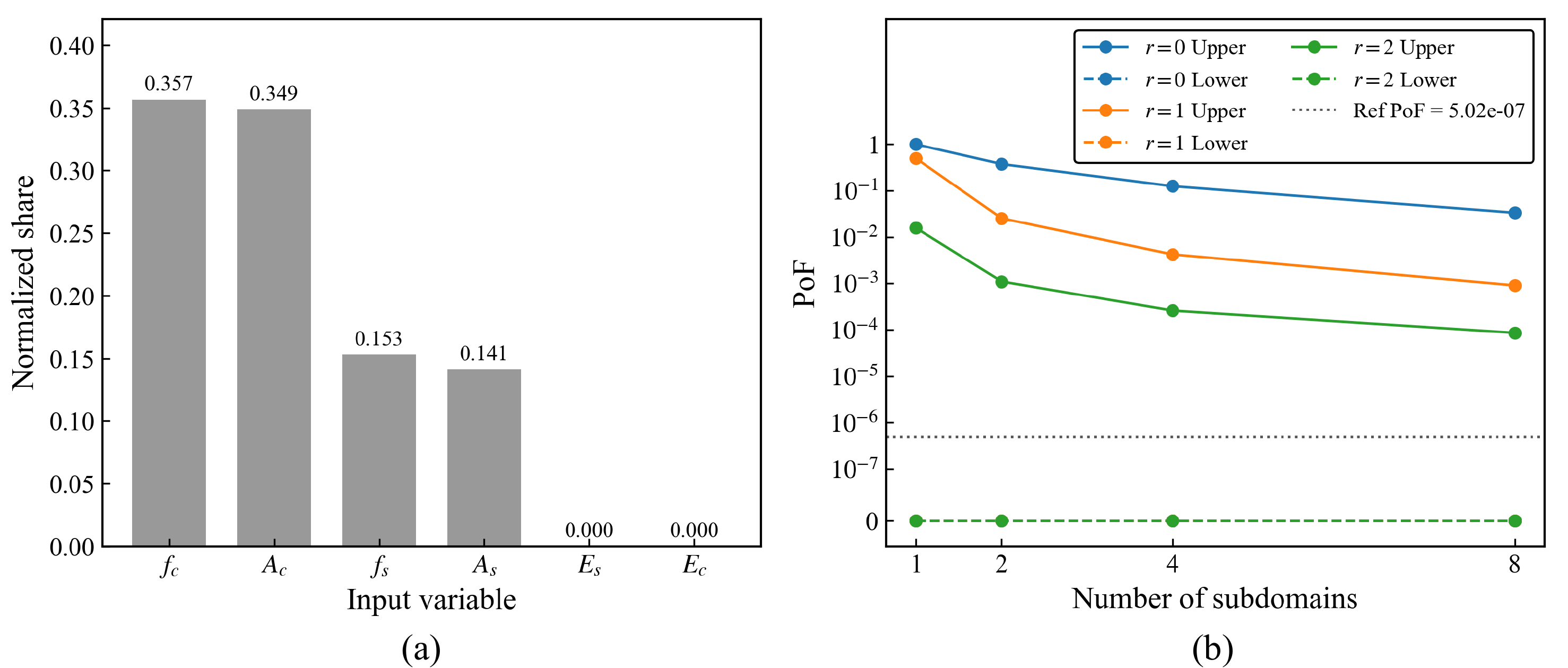}
\caption{(a) Ranking of normalized McDiarmid subdiameters for the six candidate variables $(A_s,A_c,E_s,E_c,f_s,f_c)$. (b) Optimal PoF bounds for the reduced four-variable active set $(q,l,A_c,f_c)$ under different combinations of $K$ and $r$.}
\label{fig:8d_ouq_compare}
\end{figure}

\begin{table}[width=0.8\linewidth,cols=4,pos=htbp]
\caption{Results of McDiarmid subdiameters for $(A_s,A_c,E_s,E_c,f_s,f_c)$.}
\label{tab:McDiarmid_results}
\begin{tabular*}{\tblwidth}{@{} L L L @{}}
\toprule
Parameter & McDiarmid subdiameter & Normalized share \\
\midrule
Cross-sectional area $A_s$ & $1.1404\times10^5$ & $0.1413$ \\
Cross-sectional area $A_c$ & $2.8169\times10^5$ & $0.3490$\\
Young's modulus $E_s$ & $9.7675\times10^{-3}$ & $0.0000$ \\
Young's modulus $E_c$ & $7.1947\times10^{-4}$ & $0.0000$ \\
Tensile strength $f_s$ & $1.2363\times10^5$ & $0.1532$ \\
Compressive strength $f_c$ & $2.8782\times10^5$ & $0.3566$ \\
\bottomrule
\end{tabular*}
\end{table}

Consistent with the general trend observed in the previous examples, increasing either the number of subdomains or the number of moment constraints progressively tightens the OUQ bounds. When the available uncertainty information is minimal, namely for $K=1$ and $r=0$, the optimal bounds remain $[L,U]=[0,1]$, providing no useful certification information. When only zeroth-order subdomain information is incorporated, the upper bound decreases from $3.75\times10^{-1}$ at $K=2$ and $r=0$ to $1.25\times10^{-1}$ at $K=4$ and $r=0$. When first-order moment constraints are additionally imposed, the upper bound is further reduced to $2.5467\times10^{-2}$ at $K=2$ and $r=1$, $4.2388\times10^{-3}$ at $K=4$ and $r=1$, and $9.0931\times10^{-4}$ at $K=8$ and $r=1$. With second-order moment constraints, the upper bound decreases more substantially, reaching $1.1072\times10^{-3}$ at $K=2$ and $r=2$, $2.6038\times10^{-4}$ at $K=4$ and $r=2$, and finally $8.6255\times10^{-5}$ at $K=8$ and $r=2$. In contrast, for all cases considered here, the computed lower bounds remain numerically equal to zero.

These results indicate that, even when uncertainty information is enriched only on a subset of active dimensions, the proposed subdomain-based OUQ framework remains capable of substantially compressing the PoF upper bound in this rare-event regime. Relative to the subset-simulation reference PoF $5.0249\times10^{-7}$, the upper bound is reduced by more than four orders of magnitude, from the vacuous value $1$ at $K=1$ and $r=0$ to $8.6255\times10^{-5}$ at $K=8$ and $r=2$. This demonstrates that active-dimension refinement provides a practical means of controlling the computational cost while preserving the essential tightening behavior of the OUQ bounds. However, the upper bound \(8.6255\times10^{-5}\) obtained with the richest uncertainty information considered here, \(K=8\) and \(r=2\), remains about two orders of magnitude larger than the reference PoF \(5.0249\times10^{-7}\). This is mainly due to the fact that the nonactive variables \((A_s,E_s,E_c,f_s)\) are kept at the coarsest resolution, namely \(K=1\) and \(r=0\).

A further observation from Fig.~\ref{fig:8d_ouq_compare}(b) and Table~\ref{tab:ouq_case4_rooftruss_bounds} is that the upper bound is more sensitive to the order of the moment constraints than to the number of subdomains. This can be seen by comparing the cases \(K=4\), \(r=0\) and \(K=2\), \(r=1\), which have the same computational cost in terms of objective function evaluations. However, the resulting upper bounds differ substantially: the upper bound for \(K=2\), \(r=1\) is \(0.0255\), which is one order of magnitude smaller than that for \(K=4\), \(r=0\), namely \(0.1250\). A similar trend is observed by comparing \(K=8\), \(r=0\) with \(K=4\), \(r=1\). These results suggest that, for this roof-truss problem, improving bound sharpness is more effectively achieved by incorporating higher-order local moment information than by merely increasing the number of subdomains.

It is also worth noting that the lower bounds remain unchanged at zero despite the addition of increasingly informative local constraints. This behavior is consistent with the rare-event nature of the present benchmark and indicates that, under the current information budget and optimization tolerance, the available uncertainty information is still insufficient to raise the lower bound away from the numerical zero level. Consequently, for this class of rare-event problems, the main practical benefit of the current framework lies in progressively tightening conservative upper bounds rather than in simultaneously producing informative two-sided certification intervals.

\subsection{Case 5: ATen-Dimensional Ballistic Impact Problem}
\label{sec:ballistic}

We next illustrate the developed OUQ framework through an application involving the performance of an AZ31B magnesium alloy plate subjected to normal impact by a heavy steel ball, as shown in Fig.~\ref{fig:impact_geo}(a). The plate measures $10$~cm in both length and width, with a thickness of $0.35$~cm, while the ball has a diameter of $1.12$~cm. The design specification is defined in terms of the maximum backface deflection of the plate $Y$, illustrated in Fig.~\ref{fig:impact_geo}(b). The system is considered safe if the maximum backface deflection remains below a prescribed threshold $Y^{\mathrm{c}}$. Otherwise, the design is deemed to have failed. We further assume that all uncertainty arises from imperfect characterization of the material variables and impact conditions. Specifically, the ten uncertain inputs consist of nine parameters associated with the target material model and the impact velocity. Partial information about the input probability measure is assumed to be available, specifically in the form of statistical moment constraints defined over subdomains. This allows the proposed OUQ framework to be examined in a substantially higher-dimensional setting. For simplicity, the projectile is assumed to be free of uncertainty. 

\begin{figure}[pos=htbp]
\centering
\includegraphics[width=\textwidth]{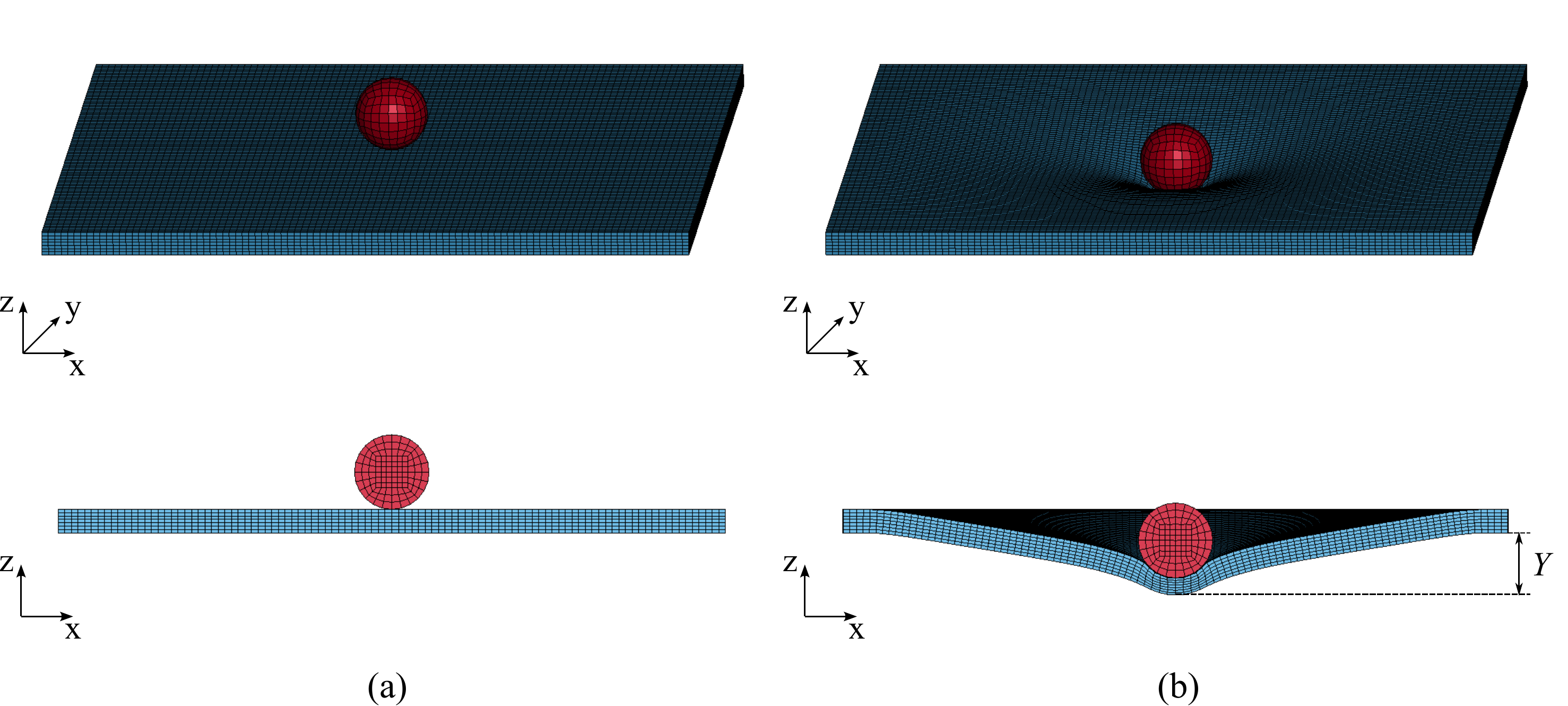}
\caption{Schematic illustration of an AZ31B magnesium plate impacted by a spherical steel projectile: (a) Initial configuration, and (b) Performance metric, with the maximum backface deflection denoted by $Y$. In each subfigure, the top figure shows a three-dimensional view of the projectile-plate system, while the bottom figure presents the corresponding mid-plane $x$-$z$ cross-sectional view. Red and green markers denote the finite element models of the projectile and the plate, respectively.}
\label{fig:impact_geo}
\end{figure}

\subsubsection{Material Modeling}

The constitutive behavior of the plate is described by a Johnson-Cook(JC) plasticity model~\cite{johnson1983constitutive}, which is widely used to characterize the elastic-plastic response of ductile metals over broad ranges of strain rates and temperatures. Within this model, the flow stress is decomposed into contributions from strain hardening, strain-rate sensitivity, and thermal softening, and is expressed as
\begin{equation}
  \sigma_{\mathrm{y}} = 
  \bigl[A + B\,\epsilon_{\mathrm{pl}}^{\,n}\bigr]
  \bigl[1 + C \ln \dot{\epsilon}_{\mathrm{pl}}^\ast \bigr]
  \bigl[1 - T^{\ast m}\bigr],
\label{eq:JCplasticity}
\end{equation}
where $\sigma_{\mathrm{y}}$ denotes the true flow stress and $\epsilon_{\mathrm{pl}}$ is the equivalent plastic strain. The normalized plastic strain rate $\dot{\epsilon}_{\mathrm{pl}}^\ast$ and the normalized temperature $T^\ast$ are defined as
\begin{equation}
  \dot{\epsilon}_{\mathrm{pl}}^\ast = \frac{\dot{\epsilon}_{\mathrm{pl}}}{\dot{\epsilon}_{\mathrm{pl0}}}, 
  \quad
  T^\ast = \frac{T - T_0}{T_{\mathrm{m}} - T_0},
\label{eq:JCeStarTstar}
\end{equation}
where $\dot{\epsilon}_{\mathrm{pl}}$ is the plastic strain rate, $\dot{\epsilon}_{\mathrm{pl0}}$ is a reference plastic strain rate, $T$ is the temperature, $T_0$ is the room temperature, and $T_{\mathrm{m}}$ is the melting temperature of the material. In Eq.~(\ref{eq:JCplasticity}), $A$ denotes the quasi-static yield stress, $B$ the hardening modulus, $n$ the strain-hardening exponent, $C$ the strain-rate sensitivity coefficient, and $m$ the thermal softening exponent.

The volumetric response of the target plate is modeled using the Mie-Grüneisen equation of state (EOS)~\cite{gruneisen1912theorie}. The pressure is written differently in the compressed and expanded states. Let $\mu=\rho/\rho_0-1$ denote the nominal volumetric compression, where $\rho$ and $\rho_0$ are the current and reference mass densities, respectively, and let $e_0$ denote the specific internal energy per unit mass. For the compressed state, the pressure is given by
\begin{equation}
    p=
    \frac{\rho_0 c_0^2\,\mu\bigl[1+\tfrac{1}{2}(1-\gamma_0)\mu-\tfrac{\alpha}{2}\mu^2\bigr]}
         {\bigl[1-(s_1-1)\mu\bigr]^2}
    +\bigl(\gamma_0+\alpha\mu\bigr)e_0,
\label{eq:MG_EOS_comp}
\end{equation}
whereas for the expanded state it reduces to
\begin{equation}
    p=\rho_0 c_0^2\,\mu+\bigl(\gamma_0+\alpha\mu\bigr)e_0.
\label{eq:MG_EOS_exp}
\end{equation}
Here, $c_0$ is the zero-pressure bulk sound speed, $\gamma_0$ is the Gr\"uneisen gamma, and $\alpha$ is the first-order volume correction to $\gamma_0$. The coefficients $c_0$ and $s_1$ are related to a linear shock-particle velocity relation
\begin{equation}
    v_s=c_0+s_1 v_p,
\label{eq:UsUp}
\end{equation}
where $v_s$ and $v_p$ are the shock and particle velocities, respectively.

In the present study, the random input vector is taken as
\begin{equation}
\label{eq:10d_input}
    X \equiv (A,B,n,C,m,v_0,E,c_0,\gamma_0,s_1),
\end{equation}
where $(A,B,n,C,m)$ are the JC plasticity parameters, $v_0$ is the impact velocity, $E$ is the Young's modulus of the target plate, and $(c_0,\gamma_0,s_1)$ are the Mie-Gr\"uneisen EoS parameters. These inputs are assumed to be mutually independent and uniformly distributed within prescribed bounds. The reference values and bounds are summarized in Table~\ref{tab:AZ31B_10d_bounds}. Each uncertain input is assigned a bounded interval defined by $\pm 10\%$ of its reference value.

The Taylor-Quinney factor is employed to quantify the fraction of plastic work converted into heat during plastic deformation~\cite{taylor1934latent}. Due to the substantial disparity in stiffness between AZ31B magnesium and steel, the projectile is assumed to behave as a rigid body. All remaining fixed material parameters for the target plate and the projectile are summarized in Tables~\ref{tab:target_mat} and~\ref{tab:project_mat}, respectively.


\begin{table}[width=.8\linewidth,cols=6,pos=htbp]
\caption{Reference values and bounds of the uncertain inputs used in the ballistic impact problem. The lower and upper bounds are defined as $\pm 10\%$ of the corresponding reference values.}
\label{tab:AZ31B_10d_bounds}
\begin{tabular*}{\tblwidth}{@{} L L L L L L@{}}
\toprule
Parameter & Reference value & Lower bound & Upper bound & Unit & Source\\
\midrule
Johnson Cook $A$         & $225.171$ & $202.654$ & $247.688$ & [MPa] & \cite{hasenpouth2010tensile}\\
Johnson Cook $B$         & $168.346$ & $151.511$ & $185.181$ & [MPa] & \cite{hasenpouth2010tensile}\\
Johnson Cook $n$         & $0.242$   & $0.2178$  & $0.2662$  & - & \cite{hasenpouth2010tensile}\\
Johnson Cook $C$         & $0.013$   & $0.0117$  & $0.0143$  & - & \cite{hasenpouth2010tensile}\\
Johnson Cook $m$         & $1.550$   & $1.395$   & $1.705$   & - & \cite{hasenpouth2010tensile}\\
Impact velocity $v_0$       & $200.0$   & $180.0$   & $220.0$   & [m/s] & -\\
Young's Modulus $E$         & $45.0$    & $40.5$    & $49.5$    & [GPa] & \cite{hasenpouth2010tensile}\\
Grüneisen intercept $c_0$       & $4520.0$  & $4068.0$  & $4972.0$  & [m/s] & \cite{feng2017numerical}\\
Grüneisen gamma $\gamma_0$  & $1.54$    & $1.386$   & $1.694$   & - & \cite{feng2017numerical}\\
Grüneisen slope $s_1$       & $1.242$   & $1.1178$  & $1.3662$  & - & \cite{feng2017numerical}\\
\bottomrule
\end{tabular*}
\end{table}

\begin{table}[width=.8\linewidth,cols=4,pos=h]
\caption{Fixed material parameters for the AZ31B target plate.}
\label{tab:target_mat}
\begin{tabular*}{\tblwidth}{@{} Y{6cm} L L L @{}}
\toprule
Parameter & Value & Unit & Source \\
\midrule
Mass density  & $1.77$    & [g/cm$^3$]    & \cite{hasenpouth2010tensile} \\
Poisson's ratio         & $0.35$    & -             & - \\
Specific heat capacity  & $1.005$   & [J/(K·g)]     & \cite{lee2013thermal} \\
Taylor-Quinney factor   & $0.6$     & -             & \cite{kingstedt2019conversion} \\
Reference strain rate   & $0.001$   & [s$^{-1}$]    & \cite{hasenpouth2010tensile} \\
Reference Temperature   & $298.15$  & [K]           & \cite{hasenpouth2010tensile} \\
Reference melting Temperature & $905.0$   & [K]     & \cite{hasenpouth2010tensile}\\
\bottomrule
\end{tabular*}
\end{table}

\begin{table}[width=.8\linewidth,cols=3,pos=h]
\caption{Fixed material parameters for the steel projectile.}
\label{tab:project_mat}
\begin{tabular*}{\tblwidth}{@{} Y{3cm} L L @{}}
\toprule
Parameter & Value & Unit \\
\midrule
Mass density  & $7.83$    & [g/cm$^3$]     \\
Young's Modulus         & $210.0$   & [GPa]       \\
Poisson's ratio         & $0.30$    & -          \\
\bottomrule
\end{tabular*}
\end{table}

\subsubsection{Forward Solver and Surrogate Modeling}

For a given realization of the input parameters, the maximum backface deflection $Y$ is computed using an explicit finite-element solver implemented in LS-DYNA~\cite{lsdynatheory}. The projectile is discretized using $864$ finite elements, while the plate is modeled with $70,000$ elements. As shown in Fig.~\ref{fig:impact_geo}, all elements are linear hexahedral elements with single-point integration and appropriate hourglass control. The backface nodes of the target plate near the edges are fully constrained to prevent displacement in all directions. Normal contact between the projectile and the plate is enforced using a penalty-based contact formulation~\cite{lsdynatheory}. The time step is governed by the critical element size according to the Courant-Friedrichs-Lewy (CFL) condition~\cite{courant1928partiellen}, and the total simulation time is set to $500~\mu\mathrm{s}$, which is sufficient to capture the complete rebound of the projectile. All simulations are performed under adiabatic conditions, with the initial temperature set to room temperature.

Direct use of the LS-DYNA solver within the objective function of the OUQ optimization is computationally prohibitive. Therefore, we construct a data-driven neural-network surrogate $\mathcal{N}_\theta$, parameterized by network parameters $\theta$, to approximate the forward mapping from the input vector $X = (A,B,n,C,m,v_0,E,c_0,\gamma_0,s_1)$ to the maximum back-face deflection $Y$, i.e.,
\begin{equation}
    Y = G(X) \approx \mathcal{N}_{\theta}(X).
\end{equation}

A dataset of $10{,}000$ samples is generated using the LS-DYNA solver with Latin hypercube sampling (LHS) over the input space defined in Table~\ref{tab:AZ31B_10d_bounds}. The dataset is randomly partitioned into training, validation, and test sets with fixed proportions of $70\%$, $20\%$, and $10\%$, respectively~\cite{bishop2006pattern}. The training set is used to fit the network parameters, the validation set is used for early stopping to prevent overfitting, and the test set is reserved for assessing the surrogate's generalization performance.

The surrogate model is implemented as a fully connected multilayer perceptron (MLP) with four hidden layers of $256$ units each, yielding the architecture $10 \to 256 \to 256 \to 256 \to 256 \to 1$. Scaled exponential linear unit (SELU) activations~\cite{klambauer2017self} are employed in all hidden layers, with a linear activation at the output. The network is trained using the mean-squared error (MSE) loss and the Adam optimizer~\cite{kingma2014adam} with a learning rate of $10^{-3}$, using mini-batches of size $64$. Training is performed for up to $500$ epochs with early stopping~\cite{prechelt2002early}. The validation MSE is monitored, and training is terminated if no improvement is observed for $50$ consecutive epochs. Model parameters are checkpointed whenever the validation loss reaches a new minimum, and the best-performing checkpoint is used for all subsequent evaluations.

The performance of the trained surrogate is summarized in Fig.~\ref{fig:surrogate_performance}. Fig.~\ref{fig:surrogate_performance}(a) shows the training and validation MSEs as functions of epoch on a logarithmic scale. The logarithmic axis makes it easier to visualize the optimization process over several orders of magnitude and to identify the epoch selected by early stopping. As indicated by the black marker and dashed vertical line, the minimum validation MSE is attained at epoch $204$, and the corresponding checkpoint is adopted as the final surrogate model. Although both curves exhibit small fluctuations during training, they remain at low levels after the initial decay, indicating stable optimization without evidence of sustained overfitting. Fig.~\ref{fig:surrogate_performance}(b) presents the parity plot on the test set, where the surrogate predictions are compared against the LS-DYNA outputs for the maximum backface deflection. Most samples cluster tightly around the reference line $y=x$, confirming that the surrogate captures the input-output mapping with high fidelity over the test domain. Quantitatively, the test-set coefficient of determination is $R^2=0.9998$, with a root mean squared error (RMSE) of $6.74\times10^{-4}$~cm and a mean absolute error (MAE) of $4.04\times10^{-4}$~cm. A few isolated points deviate more visibly from the diagonal than the majority of the samples; however, these deviations are sparse and do not affect the overall near-identity trend of the parity plot. Taken together, the loss history and the test-set parity results indicate that the surrogate provides an accurate and sufficiently robust approximation of the LS-DYNA forward model for the subsequent OUQ computations.

\begin{figure}[pos=htbp]
\centering
\includegraphics[width=0.9\textwidth]{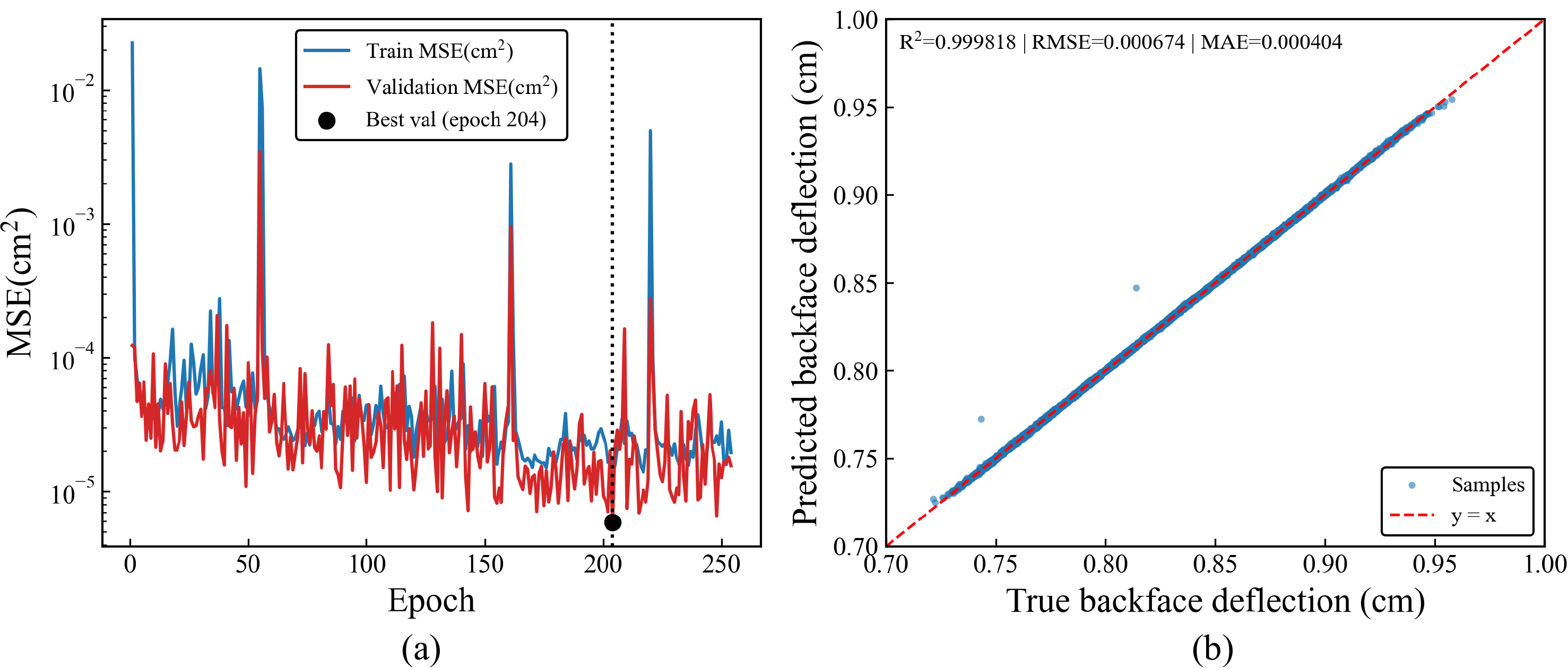}
\caption{Surrogate model performance: (a) training and validation MSE histories on a logarithmic scale, with the selected early-stopping checkpoint marked at epoch $204$; and (b) parity plot comparing surrogate predictions with LS-DYNA outputs on the test set.}
\label{fig:surrogate_performance}
\end{figure}

To quantify the computational acceleration enabled by the surrogate forward model, we benchmark its performance against the LS-DYNA solver on the same workstation using a single CPU core. A single LS-DYNA impact simulation requires approximately $649$~s. By contrast, the trained surrogate model, evaluated in CPU inference mode and averaged over repeated calls after a warm-up stage to suppress timing noise, requires only $7.73 \times 10^{-5}$~s per sample. This yields an acceleration factor of approximately $8.39 \times 10^{6}$, corresponding to nearly seven orders of magnitude. Since OUQ and Monte Carlo analyses typically evaluate the forward model in batches, the achievable throughput can be even higher in practice. For a batch size of $1024$, the effective average inference time is reduced to $3.65 \times 10^{-6}$~s per sample, corresponding to a speedup of eight orders of magnitude. Hence, the single-sample benchmark reported here should be regarded as a conservative lower bound on the computational gain afforded by the surrogate model in large-scale UQ calculations.

\subsubsection{OUQ Results Under Varying Levels of Uncertainty Information}

We simultaneously partition the domains of the ten uncertain input parameters into $K = 1$, $2$, $4$, and $8$ subdomains, and impose moment constraints up to second order, i.e., $r = 0$, $1$, and $2$. As in Case 2, whenever exact enumeration of all joint Dirac combinations is not feasible, the PoF in the present ballistic-impact example is evaluated using the same ITS-based Monte Carlo strategy, with $N_{\mathrm{ITS}}=5\times10^4$ samples. We first consider the failure threshold $Y^{\mathrm{c}} = 0.93$~cm. Fig.~\ref{fig:10d_impact_compare}(a) compares the resulting optimal bounds with the reference PoF, which is estimated via Monte Carlo sampling of the surrogate ballistic-impact model. The numerical values of the optimal bounds are reported in Table~\ref{tab:ouq_10d_impact_bounds} in Appendix~\ref{sec:tables}. Consistent with earlier examples, increasing either the number of subdomains or the number of moment constraints reduces the upper bound $U$ and increases the lower bound $L$, thereby progressively tightening the OUQ interval. When the available uncertainty information is minimal---specifically in the case $K=1$, $r=0$---the optimal bounds are $[L,\,U]=[0,\,1]$, providing no useful certification information. In contrast, when the uncertainty description is most informative, namely for $K=8$ and $r=2$, the optimal bounds tighten to $[L,\,U]=[0.0131,\,0.0136]$, deviating from the reference PoF of $0.0134$ by only $[-1.62\%,\, +1.75\%]$.

A further observation from Fig.~\ref{fig:10d_impact_compare}(a) is that subdomain partitioning and moment constraints exhibit different efficiency-tightness trade-offs. When only subdomain information is imposed, the bounds can be tightened to $[L,\,U]=[0.0003,\,0.0485]$ at $K=8$, corresponding to deviations of $[-97.99\%,\, +263.39\%]$ from the reference PoF. Conversely, when only moment constraints are imposed, the bounds tighten to $[L,\,U]=[0,\,0.0892]$ at $r=2$, yielding much larger deviations of $[-100.00\%,\, +568.03\%]$. Although subdomain-only constraints produce tighter bounds than moment-only constraints, they also require a higher computational cost.

Despite these differences, Fig.~\ref{fig:10d_impact_compare}(a) also indicates that, when computational cost is taken into account, higher-order moment information plays a more important role than finer subdomain partitioning in improving the tightness of the optimal bounds for the present problem. For example, the bounds obtained for $K=4$, $r=0$, namely $[L,\,U]=[0,\,0.1144]$, are slightly wider than those obtained for $K=2$, $r=1$, which yield $[L,\,U]=[0,\,0.0715]$. A similar pattern is observed when comparing cases $K=8$, $r=0$ and $K=4$, $r=1$. These results suggest that, for this application, comparable improvements in PoF bound tightness can be achieved either by refining the subdomain partitioning or by enriching the uncertainty description through higher-order moment information across the input parameters.

\begin{figure}[pos=htbp]
\centering
\includegraphics[width=0.9\textwidth]{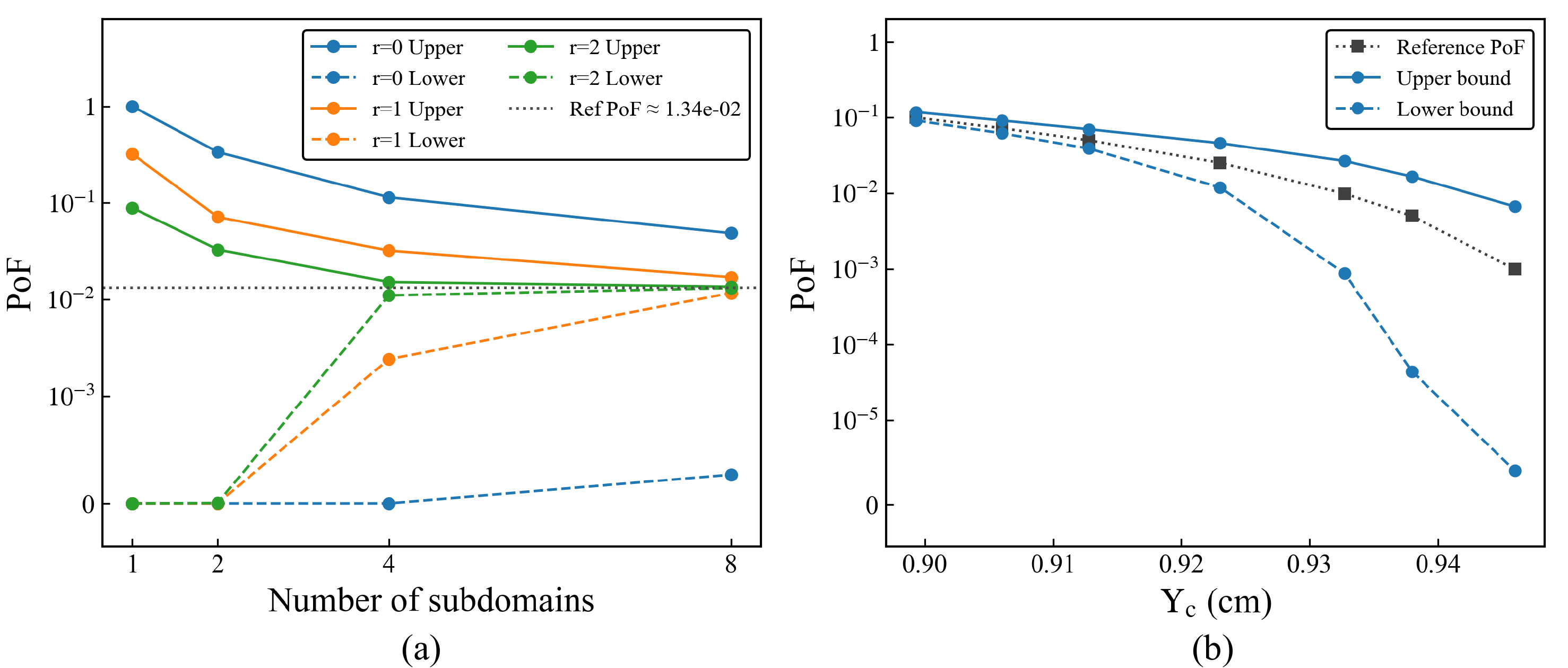}
\caption{OUQ bounds for the ten-dimensional ballistic-impact problem: (a) Optimal bounds with different $K$ and $r$ for $Y^{\mathrm{c}}=0.93$~cm; and (b) Optimal bounds for fixed $K=4$ and $r=1$ as $Y^{\mathrm{c}}$ varies.}
\label{fig:10d_impact_compare}
\end{figure}

\subsubsection{OUQ Results Under Varying Failure Thresholds}

To further examine the performance of the proposed framework as the failure event becomes rarer, we fix $K=4$ and $r=1$ and vary the failure threshold $Y^{\mathrm{c}}$. Fig.~\ref{fig:10d_impact_compare}(b) compares the resulting optimal bounds with the corresponding reference PoF. Table~\ref{tab:ouq_10d_impact_yc} in Appendix~\ref{sec:tables} reports the explicit numerical values of the reference PoFs and their corresponding optimal bounds. As $Y^{\mathrm{c}}$ increases, the reference PoF decreases monotonically from $9.9896\times10^{-2}$ to $9.86\times10^{-4}$, and both optimal bounds decrease accordingly. For relatively moderate failure probabilities, the bounds remain reasonably informative. For example, at $Y^{\mathrm{c}}=0.8993$, the optimal bounds are $[L,\,U]=[0.0916,\,0.1181]$, deviating from the reference PoF of $0.099896$ by $[-8.36\%,\, +18.23\%]$.

As the failure event becomes rarer, however, the optimal interval becomes increasingly conservative in relative terms. At $Y^{\mathrm{c}}=0.9327$, for which the reference PoF is $9.88\times10^{-3}$, the bounds are $[L,\,U]=[8.72\times10^{-4},\,2.67\times10^{-2}]$. When $Y^{\mathrm{c}}=0.9460$, the reference PoF further decreases to $9.86\times10^{-4}$, while the optimal bounds become $[L,\,U]=[4\times10^{-6},\,6.67\times10^{-3}]$. These results indicate that the proposed subdomain-based OUQ framework still captures the correct order of magnitude of the PoF as the event becomes rare, but the certification interval widens substantially in relative terms, especially because the lower bound rapidly approaches zero. This behavior highlights the increasing difficulty of sharply localizing extremal admissible measures in the tail region when only limited subdomain-based moment information is available.

\section{Concluding Remarks}
\label{sec:summary}

In this study, we have formulated an OUQ framework to account for uncertain inputs characterized by truncated moment constraints defined over input subdomains. We have developed a high-performance computational framework to compute rigorous optimal upper and lower bounds on the PoF over the admissible set of probability measures. The proposed methodology has integrated several key components: a reduction theory that transforms the original infinite-dimensional optimization problem over admissible probability measures into a finite-dimensional problem over admissible Dirac measures; the use of canonical moments that converts constrained optimization problems into unconstrained ones; eigendecomposition of the associated Jacobi matrix that mitigates ill-conditioning in root-finding procedures, and ITS that enables efficient and accurate evaluation of the PoF under Dirac measures.

We have presented five groups of numerical examples to evaluate the effectiveness of the proposed subdomain-based OUQ framework across varying numbers of subdomains and moment constraints. Based on these investigations, the following conclusions can be drawn:
\begin{enumerate}
    \item The evidence theory can be regarded as a special case of the proposed OUQ methodology, in which the uncertain inputs are constrained only by zeroth-order moments over subdomains. In this setting, the plausibility and belief functions in the evidence theory correspond, respectively, to the optimal upper and lower bounds obtained from the OUQ.

    \item For high-dimensional problems, the ITS strategy has reduced computational cost by up to two orders of magnitude while maintaining a relative error below \(1\%\). By contrast, for low-dimensional problems, the cost of ITS may exceed that of directly evaluating the PoF using all Dirac combinations.

    \item Across all numerical examples considered, a general trend is observed in which increasing either the number of subdomains or the order of moment constraints has systematically decreased the upper bound and increased the lower bound, thereby tightening the bound interval. In all cases, the true/reference PoF is consistently bracketed by the corresponding optimal bounds.

    \item The proposed framework has been demonstrated on smooth and non-smooth examples, rare events, and high-dimensional engineering problems with up to ten random inputs. For rare events, it remains informative, although certification becomes more difficult as the target PoF decreases. In such cases, the main practical value is the progressive tightening of rigorous upper bounds, which can serve as conservative safety-design criteria.

    \item We have identified regimes and input parameters for which the optimal bounds are more sensitive to subdomain partitioning or to higher-order moment information. In some cases, localized subdomain information is more effective for tightening the OUQ bounds, whereas in others higher-order moments provide a more favorable improvement when computational cost is taken into account. These findings provide practical guidance for prioritizing uncertainty reduction efforts when certifying system safety.
\end{enumerate}

In this work, we have focused on computing optimal PoF bounds for systems whose uncertain inputs are characterized by truncated moment constraints over input subdomains, a form of epistemic uncertainty. In many engineering applications, however, sufficient data may be available for certain uncertain quantities to justify probabilistic modeling using stochastic distributions, while only limited data exist for other quantities, precluding such representations. An important direction for future research is therefore the integration of aleatory uncertainties, described by stochastic distributions, with epistemic uncertainties characterized by moment constraints over subdomains, within a unified OUQ framework~\cite{miska2022method}. 

It is also worth noting that, in some instances, the upper or lower bounds remain unchanged despite the introduction of additional uncertainty information. This behavior arises when the added information pertains to regions of the input space that lie entirely within the safe and failure domain and therefore do not influence the bounds on PoF, which are primarily governed by probability mass near the boundary between the safe and failure domains. Consequently, uncertainty information provided on subdomains that intersect, but are not entirely contained within, the failure region is more effective for tightening the PoF bounds. These marginal subdomains may further benefit from finer partitioning or the incorporation of higher-order moment constraints to improve bound tightness more efficiently. Accordingly, the identification of failure regions in the input space and the development of adaptive subdomain partitioning strategies constitute promising directions for future research.

It should be noted that the affine transformation in Eq.~\eqref{eq:affine_transform} requires each subdomain \(\mathcal{X}_{i,j}=[a_{i,j},b_{i,j}]\) to be bounded. Therefore, the present canonical-moment construction is formulated for bounded input domains. In many engineering applications, uncertain variables have natural physical or design bounds. When such bounds are not available, practical bounds may be specified based on expert judgment, available data, engineering constraints, or high-probability truncation of a reference distribution, as adopted in the numerical examples. The resulting OUQ bounds should then be interpreted relative to the prescribed bounded domain. Extending the present formulation to genuinely unbounded supports would require a different moment-parameterization strategy and is left for future work.

The present study focuses on statistically independent input variables, and its direct applicability to dependent inputs is therefore limited. More general UQ methodologies for dependent random variables have been developed in the literature, particularly in polynomial-chaos- and polynomial-dimensional-decomposition-based settings when joint or marginal dependence information is available. Representative examples include convergence theory for generalized polynomial chaos expansions under arbitrary probability measures~\cite{ernst2012convergence}, moment-based construction of multivariate polynomial chaos bases for dependent variables~\cite{feinberg2018multivariate}, measure-consistent polynomial chaos expansions for statistically dependent inputs~\cite{rahman2018polynomial}, and generalized polynomial dimensional decomposition under dependent random variables~\cite{rahman2019uncertainty}. These studies provide important context for positioning the present work relative to dependent-input UQ. Extending the present subdomain-moment OUQ framework to dependent inputs is possible in principle within the broader OUQ philosophy, but it would require a different admissible-set construction based on joint measures, copulas, or other dependence descriptions, together with a corresponding reformulation of the reduction and numerical solution procedures. Such an extension is therefore an important direction for future research.

We close this section by discussing the application of the proposed OUQ framework to high-dimensional rare-event problems. The present work focuses on either low-dimensional, low-probability cases (e.g., four inputs with PoF on the order of $10^{-7}$) or high-dimensional, moderate-probability cases (e.g., ten inputs with PoF on the order of $10^{-2}$). Direct application to high-dimensional, low-probability settings remains challenging. This difficulty arises because the required number of inverse-transform samples, $N_{\mathrm{ITS}}$, grows prohibitively large in high dimensions for rare events. For example, achieving reliable estimates for PoF on the order of $10^{-7}$ would typically require $N_{\mathrm{ITS}}$ on the order of $10^{9}$. This challenge is analogous to that encountered in standard Monte Carlo methods. Consequently, extending the OUQ framework to such regimes will likely require more advanced sampling strategies, such as importance sampling~\cite{melchers1989importance, engelund1993benchmark} or subset simulation~\cite{au2001estimation, au2003subset}, adapted to the Dirac-measure representation of the input probability space.

\section{Acknowledgments}
R.J. and X.S. gratefully acknowledge the support of the National Science Foundation (NSF), USA, under Grant No. $2429424$, and the support of the University of Kentucky, USA, through the faculty start-up fund. The authors also thank Dr. Burigede Liu for generously sharing the learning-based OUQ codes. The comments and suggestions of the anonymous reviewers have improved the quality and scope of this work.

\section*{Appendix}
\appendix

\section{Optimal Bounds under Zeroth-Order Moment Constraints}
\label{sec:deri}

As defined in Eq.~(\ref{eq:failuredomain}), the failure domain 
$\mathcal{X}^{\mathrm{c}}$ is a subset of the input space 
$\mathcal{X}$. Using Eqs.~(\ref{eq:failuredomain}) and 
(\ref{eq:domaindecomp}), we obtain
\begin{equation}
\begin{aligned}
    \mathcal{X}^{\mathrm c}
    &= \mathcal{X}^{\mathrm c}\cap\mathcal{X} \\
    &= 
    \mathcal{X}^{\mathrm c}
    \cap
    \Bigg( \bigcup_{\alpha_1=1}^{K}\cdots\bigcup_{\alpha_m=1}^{K}
    \mathcal{R}_{\alpha} \Bigg) \\
    &=
    \bigcup_{\alpha_1=1}^{K}\cdots\bigcup_{\alpha_m=1}^{K}
    \left(
    \mathcal{X}^{\mathrm c}\cap\mathcal{R}_{\alpha}
    \right).
\end{aligned}
\label{eq:xcxi}
\end{equation}
Substituting Eq.~\eqref{eq:xcxi} into Eq.~\eqref{eq:muxy} and using the disjointness in Eq.~\eqref{eq:xixj}, the PoF can be written as
\begin{equation}
\begin{aligned}
    \mu[Y \in \mathcal{Y}_{\mathrm{c}}]
    &=
    \mu\!\left[
        X \in
        \bigcup_{\alpha_1=1}^{K}\cdots\bigcup_{\alpha_m=1}^{K}
        \left(
        \mathcal{X}^{\mathrm c}\cap\mathcal{R}_{\alpha}
        \right)
    \right] \\
    &=
    \sum_{\alpha_1=1}^{K}\cdots\sum_{\alpha_m=1}^{K}
    \mu\!\left[
    X\in
    \mathcal{X}^{\mathrm c}\cap\mathcal{R}_{\alpha}
    \right].
\end{aligned}
\label{eq:law}
\end{equation}
The last equality follows from the additivity of the probability measure. Using the law of total probability, we have
\begin{equation}
\mu[Y\in\mathcal{Y}_{\mathrm c}]
=
\sum_{\alpha_1=1}^{K}\cdots\sum_{\alpha_m=1}^{K}
\mu\!\left[
X\in\mathcal{R}_{\alpha}
\right]
\,
\mu\!\left[
X\in\mathcal{X}^{\mathrm c}
\,\big|\,
X\in\mathcal{R}_{\alpha}
\right].
\end{equation}
Because the inputs are independent, we have
\begin{equation}
\mu\!\left[
X\in\mathcal{R}_{\alpha}
\right]
=
\prod_{i=1}^{m}
\mu_i\left[X_i\in\mathcal{X}_{i,\alpha_i}\right]
=
\prod_{i=1}^{m} M_{i,\alpha_i,0}.
\end{equation}
Therefore,
\begin{equation}
\mu[Y\in\mathcal{Y}_{\mathrm c}]
=
\sum_{\alpha_1=1}^{K}\cdots\sum_{\alpha_m=1}^{K}
\left(
\prod_{i=1}^{m} M_{i,\alpha_i,0}
\right)
\mu\!\left[
X\in\mathcal{X}^{\mathrm c}
\,\big|\,
X\in\mathcal{R}_{\alpha}
\right].
\end{equation}

Accordingly, the upper and lower bounds over the information set \(\mathcal{A}\) become
\begin{subequations}
\begin{equation}
U(\mathcal{A})
=
\sup_{\mu\in\mathcal{A}}
\sum_{\alpha_1=1}^{K}\cdots\sum_{\alpha_m=1}^{K}
\left(
\prod_{i=1}^{m} M_{i,\alpha_i,0}
\right)
\mu\!\left[
X\in\mathcal{X}^{\mathrm c}
\,\big|\,
X\in\mathcal{R}_{\alpha}
\right],
\end{equation}
and
\begin{equation}
L(\mathcal{A})
=
\inf_{\mu\in\mathcal{A}}
\sum_{\alpha_1=1}^{K}\cdots\sum_{\alpha_m=1}^{K}
\left(
\prod_{i=1}^{m} M_{i,\alpha_i,0}
\right)
\mu\!\left[
X\in\mathcal{X}^{\mathrm c}
\,\big|\,
X\in\mathcal{R}_{\alpha}
\right].
\end{equation}
\label{eq:opt}
\end{subequations}

In general, the supremum of a sum of functions does not equal the sum of their suprema, nor does the infimum distribute over summation. However, Proposition~\ref{prop:suminf} in Appendix~\ref{sec:thm} establishes that, for the particular structure of Eq.~\eqref{eq:opt}, the supremum and infimum do commute with the summation. Combining this result with the known zeroth moments $M_{i,j,0}$ from Eq.~\eqref{eq:infset0}, we obtain
\begin{subequations}
\begin{equation}
U(\mathcal{A})
=
\sum_{\alpha_1=1}^{K}\cdots\sum_{\alpha_m=1}^{K}
\left(
\prod_{i=1}^{m} M_{i,\alpha_i,0}
\right)
\sup_{\mu\in\mathcal{A}}
\mu\!\left[
X\in\mathcal{X}^{\mathrm c}
\,\big|\,
X\in\mathcal{R}_{\alpha}
\right],
\end{equation}
and
\begin{equation}
L(\mathcal{A})
=
\sum_{\alpha_1=1}^{K}\cdots\sum_{\alpha_m=1}^{K}
\left(
\prod_{i=1}^{m} M_{i,\alpha_i,0}
\right)
\inf_{\mu\in\mathcal{A}}
\mu\!\left[
X\in\mathcal{X}^{\mathrm c}
\,\big|\,
X\in\mathcal{R}_{\alpha}
\right].
\end{equation}
\label{eq:opt1}
\end{subequations}
Since only zeroth moments are specified, each conditional probability must satisfy
\begin{equation}
    \mu\!\left[
    X\in\mathcal{X}^{\mathrm c}
    \,\big|\,
    X\in\mathcal{R}_{\alpha}
    \right]
    =
    \begin{cases}
    0, & \mathcal{X}^{\mathrm c}\cap\mathcal{R}_{\alpha}=\varnothing,\\[4pt]
    1, & \mathcal{X}^{\mathrm c}\cap\mathcal{R}_{\alpha}=\mathcal{R}_{\alpha},\\[4pt]
    \in(0,1), & \text{otherwise}.
    \end{cases}
\end{equation}
Thus, the extremal conditional probabilities are
\begin{subequations}
\begin{equation}
    \sup_{\mu\in\mathcal{A}}
    \mu\!\left[
    X\in\mathcal{X}^{\mathrm c}
    \,\big|\,
    X\in\mathcal{R}_{\alpha}
    \right]
    =
    \mathbf{1}\!\left[
    \mathcal{X}^{\mathrm c}\cap\mathcal{R}_{\alpha}\neq\varnothing
    \right],
\end{equation}
and
\begin{equation}
    \inf_{\mu\in\mathcal{A}}
    \mu\!\left[
    X\in\mathcal{X}^{\mathrm c}
    \,\big|\,
    X\in\mathcal{R}_{\alpha}
    \right]
    =
    \mathbf{1}\!\left[
    \mathcal{X}^{\mathrm c}\cap\mathcal{R}_{\alpha}=\mathcal{R}_{\alpha}
    \right].
\end{equation}
\end{subequations}
Substituting these expressions into Eq.~\eqref{eq:opt1} yields the upper and lower bounds in Eq.~\eqref{eq:bounds0}.

\section{Theorem and Proof}
\label{sec:thm}

We begin by recalling a classical characterization of the supremum and infimum of a set of real numbers.
\begin{theorem}[Characterization of Supremum and Infimum]
\label{thm:sup}
Let $A \subset \mathbb{R}$, and suppose that both $\sup A$ and $\inf A$ exist. Then,
\begin{enumerate}[topsep=0pt, itemsep=1pt]
    \item $a = \sup A$ if and only if  
    \begin{enumerate}[topsep=0pt, itemsep=1pt]
        \item $y \le a$ for all $y \in A$, and  
        \item for every $\epsilon > 0$ there exists $x \in A$ such that $x > a - \epsilon$.
    \end{enumerate}

    \item $a = \inf A$ if and only if  
    \begin{enumerate}[topsep=0pt, itemsep=1pt]
        \item $y \ge a$ for all $y \in A$, and  
        \item for every $\epsilon > 0$ there exists $x \in A$ such that $x < a + \epsilon$.
    \end{enumerate}
\end{enumerate}
\end{theorem}

Using Theorem~\ref{thm:sup}, we now establish a key proposition used in our derivation of the optimal bounds in Section~\ref{sec:deri}.

\begin{prop}
Let $\mathcal{X}_j$ ($j=1,...,K)$ be mutually disjoint subdomains of the sample space $\mathcal{X}$ of the random variable $X$, and suppose 
$\mathcal{X}^*_j \subseteq \mathcal{X}_j$ for all $j$. Let $\mathcal{A}$ be the set of all probability measures $\mu$ on $\mathcal{X}$ satisfying the constraints
\[
\mathcal{A} = \big\{\mu : \mu \big[X \in\mathcal{X}_j \big] = M_j,~j=1,\dots,K\big\}.
\]
Then,
\begin{align}
\sup_{\mu \in \mathcal{A}} \sum_{j=1}^{K} \mu \big[X \in \mathcal{X}^*_j \big]
&= \sum_{j=1}^{K} \sup_{\mu \in \mathcal{A}} \mu \big[X \in \mathcal{X}^*_j \big], 
\label{eq:prop_sup} \\
\inf_{\mu \in \mathcal{A}} \sum_{j=1}^{K} \mu \big[X \in \mathcal{X}^*_j \big]
&= \sum_{j=1}^{K} \inf_{\mu \in \mathcal{A}} \mu \big[X \in \mathcal{X}^*_j \big].
\label{eq:prop_inf}
\end{align}
\label{prop:suminf}
\end{prop}

\begin{proof} We prove the supremum identity. The infimum case is analogous.

For any $\mu \in \mathcal{A}$ and every $j$,
\begin{equation}
\mu \big[X \in \mathcal{X}^*_j \big]
\le 
\sup_{\nu \in \mathcal{A}} \nu \big[X \in \mathcal{X}^*_j \big].
\end{equation}
Summing over $j$ gives
\begin{equation}
\sum_{j=1}^{K} \mu \big[X \in \mathcal{X}^*_j \big]
\le 
\sum_{j=1}^{K} \sup_{\nu \in \mathcal{A}} \nu \big[X \in \mathcal{X}^*_j \big].
\end{equation}
Taking the supremum over $\mu \in \mathcal{A}$ yields
\begin{equation}
\sup_{\mu \in \mathcal{A}} \sum_{j=1}^{K} \mu \big[X \in \mathcal{X}^*_j \big]
\le 
\sum_{j=1}^{K} \sup_{\nu \in \mathcal{A}} \nu \big[X \in \mathcal{X}^*_j \big].
\label{eq:upper_step}
\end{equation}
For each fixed $j$, define the constraint set
\[
\mathcal{A}_j 
:= \big\{\mu : \mu \big[X \in \mathcal{X}_j \big] = M_j \big\},
\]
i.e., the local constraint on the subdomain $\mathcal{X}_j$.  
Since $\mathcal{X}_j$'s are disjoint and $\mathcal{A}$ imposes these 
constraints independently, we have
\begin{equation}
\sup_{\mu \in \mathcal{A}} \mu \big[X \in \mathcal{X}^*_j \big]
=
\sup_{\mu \in \mathcal{A}_j} \mu \big[X \in \mathcal{X}^*_j \big].
\label{eq:Aj_equiv}
\end{equation}
By the characterization of supremum (Theorem~\ref{thm:sup}), for any $\epsilon>0$ 
there exists a measure $\mu_j^\epsilon \in \mathcal{A}_j$ such that
\begin{equation}
\mu_j^\epsilon[X \in \mathcal{X}^*_j]
>
\sup_{\nu \in \mathcal{A}_j} \nu \big[X \in \mathcal{X}^*_j \big]
- \frac{\epsilon}{K}.
\label{eq:epsilon_choice}
\end{equation}

Now define a new probability measure $\mu^\epsilon$ on $\mathcal{X}$ by
\begin{equation}
\mu^\epsilon \big[X \in \mathcal{B} \big]
=
\sum_{j=1}^K 
\mu_j^\epsilon \big[ X \in \mathcal{B} \cap \mathcal{X}_j \big],
\quad \mathcal{B} \subseteq \mathcal{X}.
\end{equation}
Since the $\mathcal{X}_j$’s are disjoint and 
\(\mu_j^\epsilon(\mathcal{X}_j)=M_j\),
this measure satisfies all constraints and thus
\begin{equation}
\mu^\epsilon \in \mathcal{A}.
\end{equation}
Then,
\begin{equation}
\mu^\epsilon \big[X \in \mathcal{X}^*_j \big]
=
\mu_j^\epsilon \big[X \in \mathcal{X}^*_j \big],
\end{equation}
and summing over \(j\), using \eqref{eq:epsilon_choice},
\begin{equation}
\sum_{j=1}^{K} \mu^\epsilon \big[X \in \mathcal{X}^*_j \big]
>
\sum_{j=1}^{K} 
\bigg(
\sup_{\nu \in \mathcal{A}_j} \nu \big[X \in \mathcal{X}^*_j \big]
- \frac{\epsilon}{K}
\bigg)
=
\sum_{j=1}^{K} \sup_{\nu \in \mathcal{A}} \nu \big[X \in \mathcal{X}^*_j \big]
- \epsilon.
\label{eq:lower_step}
\end{equation}
Because $\epsilon > 0$ is arbitrary and \(\mu^\epsilon \in \mathcal{A}\),
\begin{equation}
\sup_{\mu \in \mathcal{A}} \sum_{j=1}^{K} \mu \big[X \in \mathcal{X}^*_j \big]
\ge
\sum_{j=1}^{K} \sup_{\nu \in \mathcal{A}} \nu \big[X \in \mathcal{X}^*_j \big].
\label{eq:lower_final}
\end{equation}
 
From \eqref{eq:upper_step} and \eqref{eq:lower_final}, equality follows
\[
\sup_{\mu \in \mathcal{A}} \sum_{j=1}^{K} \mu \big[X \in \mathcal{X}^*_j \big]
=
\sum_{j=1}^{K} \sup_{\nu \in \mathcal{A}} \nu \big[X \in \mathcal{X}^*_j \big].
\]
A parallel argument proves the infimum identity. This completes the proof.
\end{proof}

\section{Numerical Values of Optimal Bounds and Relative Deviations}
\label{sec:tables}
This section reports the numerical values of the optimal OUQ bounds and their relative deviations with respect to the corresponding benchmark probabilities of failure. Results are provided for the one-dimensional examples in Tables~\ref{tab:ouq_1d_normal}-\ref{tab:ouq_1d_bimodal} (Section~\ref{sec:1D}), the five-dimensional nonlinear smooth problem in Tables~\ref{tab:ouq_5d_bounds}-\ref{tab:ouq_nits_scan} (Section~\ref{sec:5D}), the two-dimensional non-smooth four-branch problem in Tables~\ref{tab:ouq_case3a_fourbranch}-\ref{tab:ouq_case3b_fourbranch} (Section~\ref{sec:2D}), the eight-dimensional rare-event roof-truss problem in Table~\ref{tab:ouq_case4_rooftruss_bounds} (Section~\ref{sec:8D}), and the ten-dimensional ballistic impact problem in Tables~\ref{tab:ouq_10d_impact_bounds}-\ref{tab:ouq_10d_impact_yc} (Section~\ref{sec:ballistic}).

\begin{table}[width=\linewidth,cols=6,pos=h]
\caption{Optimal bounds and relative deviations for the one-dimensional truncated normal case.}
\label{tab:ouq_1d_normal}
\begin{tabular*}{\tblwidth}{@{} L L L L L L @{}}
\toprule
$r$ & $K$ & $U$ & RD of $U$ ($\%$) & $L$ & RD of $L$ ($\%$) \\
\midrule
\multirow{4}{*}{$0$}
  & $1$ & $1.0000$ & $+313.22$ & $0.0000$ & $-100.00$ \\
  & $2$ & $0.5000$ & $+106.61$ & $0.0000$ & $-100.00$ \\
  & $4$ & $0.5000$ & $+106.61$ & $0.0062$ & $-97.36$ \\
  & $8$ & $0.5000$ & $+106.61$ & $0.1056$ & $-56.36$ \\
\midrule
\multirow{4}{*}{$1$}
  & $1$ & $0.8772$ & $+262.48$ & $0.0000$ & $-100.00$ \\
  & $2$ & $0.5000$ & $+106.61$ & $0.0114$ & $-95.37$ \\
  & $4$ & $0.5000$ & $+106.61$ & $0.0261$ & $-89.21$ \\
  & $8$ & $0.4146$ & $+71.24$ & $0.1056$ & $-56.36$ \\
\midrule
\multirow{4}{*}{$2$}
  & $1$ & $0.6711$ & $+177.36$ & $0.0000$ & $-100.00$ \\
  & $2$ & $0.5000$ & $+106.61$ & $0.0128$ & $-94.71$ \\
  & $4$ & $0.4465$ & $+84.55$ & $0.0469$ & $-80.62$ \\
  & $8$ & $0.3979$ & $+64.42$ & $0.1269$ & $-47.56$ \\
\midrule
\multirow{4}{*}{$3$}
  & $1$ & $0.6349$ & $+162.36$ & $0.0042$ & $-98.26$ \\
  & $2$ & $0.4145$ & $+71.20$ & $0.0819$ & $-66.16$ \\
  & $4$ & $0.4097$ & $+69.30$ & $0.0924$ & $-61.82$ \\
  & $8$ & $0.3556$ & $+46.94$ & $0.1436$ & $-40.66$ \\
\bottomrule
\end{tabular*}

\vspace{2pt}
\parbox{\tblwidth}{%
  \footnotesize\normalfont\raggedright
  \emph{Note:} RD denotes relative deviation, computed with respect to the reference PoF of $0.2420$.%
}
\end{table}

\begin{table}[width=\linewidth,cols=6,pos=h]
\caption{Optimal bounds and relative deviations for the one-dimensional uniform case.}
\label{tab:ouq_1d_uniform}
\begin{tabular*}{\tblwidth}{@{} L L L L L L @{}}
\toprule
$r$ & $K$ & $U$ & RD of $U$ ($\%$) & $L$ & RD of $L$ ($\%$) \\
\midrule
\multirow{4}{*}{$0$}
  & $1$ & $1.0000$ & $+203.03$ & $0.0000$ & $-100.00$ \\
  & $2$ & $0.5000$ & $+51.52$  & $0.0000$ & $-100.00$ \\
  & $4$ & $0.5000$ & $+51.52$  & $0.2500$ & $-24.24$  \\
  & $8$ & $0.3750$ & $+13.64$  & $0.2500$ & $-24.24$  \\
\midrule
\multirow{4}{*}{$1$}
  & $1$ & $0.7463$ & $+126.15$ & $0.0000$ & $-100.00$ \\
  & $2$ & $0.5000$ & $+51.52$  & $0.1212$ & $-63.27$  \\
  & $4$ & $0.4338$ & $+31.45$  & $0.2500$ & $-24.24$  \\
  & $8$ & $0.3750$ & $+13.64$  & $0.2773$ & $-15.97$  \\
\midrule
\multirow{4}{*}{$2$}
  & $1$ & $0.7425$ & $+125.00$ & $0.0000$ & $-100.00$ \\
  & $2$ & $0.4951$ & $+50.03$  & $0.1237$ & $-62.52$  \\
  & $4$ & $0.4300$ & $+30.30$  & $0.2500$ & $-24.24$  \\
  & $8$ & $0.3704$ & $+12.24$  & $0.2799$ & $-15.18$  \\
\midrule
\multirow{4}{*}{$3$}
  & $1$ & $0.5966$ & $+80.79$  & $0.0982$ & $-70.24$  \\
  & $2$ & $0.4482$ & $+35.82$  & $0.1965$ & $-40.45$  \\
  & $4$ & $0.3966$ & $+20.18$  & $0.2731$ & $-17.24$  \\
  & $8$ & $0.3608$ & $+9.33$   & $0.2963$ & $-10.21$  \\
\bottomrule
\end{tabular*}

\vspace{2pt}
\parbox{\tblwidth}{%
  \footnotesize\normalfont\raggedright
  \emph{Note:} RD denotes relative deviation, computed with respect to the reference PoF of $0.3300$.%
}
\end{table}

\begin{table}[width=\linewidth,cols=6,pos=h]
\caption{Optimal bounds and relative deviations for the one-dimensional truncated Weibull case.}
\label{tab:ouq_1d_weibull}
\begin{tabular*}{\tblwidth}{@{} L L L L L L @{}}
\toprule
$r$ & $K$ & $U$ & RD of $U$ ($\%$) & $L$ & RD of $L$ ($\%$) \\
\midrule
\multirow{4}{*}{$0$}
  & $1$ & $1.0000$ & $+527.53$ & $0.0000$ & $-100.00$ \\
  & $2$ & $1.0000$ & $+527.53$ & $0.0192$ & $-87.94$  \\
  & $4$ & $0.2472$ & $+55.18$  & $0.0192$ & $-87.94$  \\
  & $8$ & $0.2472$ & $+55.18$  & $0.0767$ & $-51.86$  \\
\midrule
\multirow{4}{*}{$1$}
  & $1$ & $0.6018$ & $+277.78$ & $0.0000$ & $-100.00$ \\
  & $2$ & $0.5840$ & $+266.61$ & $0.0192$ & $-87.94$  \\
  & $4$ & $0.2472$ & $+55.18$  & $0.0574$ & $-63.98$  \\
  & $8$ & $0.2472$ & $+55.18$  & $0.0828$ & $-48.02$  \\
\midrule
\multirow{4}{*}{$2$}
  & $1$ & $0.5127$ & $+221.84$ & $0.0000$ & $-100.00$ \\
  & $2$ & $0.4376$ & $+174.70$ & $0.0192$ & $-87.94$  \\
  & $4$ & $0.2472$ & $+55.18$  & $0.0688$ & $-56.82$  \\
  & $8$ & $0.2189$ & $+37.44$  & $0.1017$ & $-36.15$  \\
\midrule
\multirow{4}{*}{$3$}
  & $1$ & $0.4632$ & $+190.78$ & $0.0034$ & $-97.87$  \\
  & $2$ & $0.3979$ & $+149.81$ & $0.0339$ & $-78.73$  \\
  & $4$ & $0.2262$ & $+41.97$  & $0.0863$ & $-45.81$  \\
  & $8$ & $0.2148$ & $+34.87$  & $0.1098$ & $-31.05$  \\
\bottomrule
\end{tabular*}

\vspace{2pt}
\parbox{\tblwidth}{%
  \footnotesize\normalfont\raggedright
  \emph{Note:} RD denotes relative deviation, computed with respect to the reference PoF of $0.1593$.%
}
\end{table}

\begin{table}[width=\linewidth,cols=6,pos=h]
\caption{Optimal bounds and relative deviations for the one-dimensional truncated bimodal normal-mixture case.}
\label{tab:ouq_1d_bimodal}
\begin{tabular*}{\tblwidth}{@{} L L L L L L @{}}
\toprule
$r$ & $K$ & $U$ & RD of $U$ ($\%$) & $L$ & RD of $L$ ($\%$) \\
\midrule
\multirow{4}{*}{$0$}
  & $1$ & $1.0000$ & $+102.96$ & $0.0000$ & $-100.00$ \\
  & $2$ & $0.6432$ & $+30.55$  & $0.0000$ & $-100.00$ \\
  & $4$ & $0.6432$ & $+30.55$  & $0.1999$ & $-59.44$  \\
  & $8$ & $0.5027$ & $+2.03$   & $0.1999$ & $-59.44$  \\
\midrule
\multirow{4}{*}{$1$}
  & $1$ & $0.8887$ & $+80.36$  & $0.0000$ & $-100.00$ \\
  & $2$ & $0.6432$ & $+30.55$  & $0.1269$ & $-74.27$  \\
  & $4$ & $0.6432$ & $+30.55$  & $0.2860$ & $-41.96$  \\
  & $8$ & $0.5027$ & $+2.03$   & $0.3488$ & $-29.21$  \\
\midrule
\multirow{4}{*}{$2$}
  & $1$ & $0.8778$ & $+78.17$  & $0.0185$ & $-96.27$  \\
  & $2$ & $0.6432$ & $+30.55$  & $0.2395$ & $-51.40$  \\
  & $4$ & $0.6198$ & $+25.80$  & $0.3113$ & $-36.82$  \\
  & $8$ & $0.5027$ & $+2.03$   & $0.4235$ & $-14.05$  \\
\midrule
\multirow{4}{*}{$3$}
  & $1$ & $0.7781$ & $+57.94$  & $0.0698$ & $-85.86$  \\
  & $2$ & $0.6281$ & $+27.49$  & $0.2673$ & $-45.76$  \\
  & $4$ & $0.6112$ & $+24.06$  & $0.3437$ & $-30.25$  \\
  & $8$ & $0.5027$ & $+2.03$   & $0.4467$ & $-9.34$   \\
\bottomrule
\end{tabular*}

\vspace{2pt}
\parbox{\tblwidth}{%
  \footnotesize\normalfont\raggedright
  \emph{Note:} RD denotes relative deviation, computed with respect to the reference PoF of $0.4927$.%
}
\end{table}

\begin{table}[width=\linewidth,cols=6,pos=h]
\caption{Optimal bounds and relative deviations for the five-dimensional nonlinear problem.}
\label{tab:ouq_5d_bounds}
\begin{tabular*}{\tblwidth}{@{} L L L L L L @{}}
\toprule
$r$ & $K$ & $U$ & RD of $U$ ($\%$) & $L$ & RD of $L$ ($\%$) \\
\midrule
\multirow{4}{*}{$0$}
  & $1$ & $1.0000$ & $+239.21$ & $0.0000$ & $-100.00$ \\
  & $2$ & $1.0000$ & $+239.21$ & $0.0000$ & $-100.00$ \\
  & $4$ & $0.8839$ & $+199.84$ & $0.0190$ & $-93.56$  \\
  & $8$ & $0.6433$ & $+118.23$ & $0.0878$ & $-70.23$  \\
\midrule
\multirow{4}{*}{$1$}
  & $1$ & $1.0000$ & $+239.21$ & $0.0000$ & $-100.00$ \\
  & $2$ & $0.6229$ & $+111.28$ & $0.1319$ & $-55.27$  \\
  & $4$ & $0.4007$ & $+35.92$  & $0.2247$ & $-23.78$  \\
  & $8$ & $0.3232$ & $+9.64$   & $0.2736$ & $-7.19$   \\
\midrule
\multirow{4}{*}{$2$}
  & $1$ & $0.6949$ & $+135.73$ & $0.0269$ & $-90.87$  \\
  & $2$ & $0.3542$ & $+20.15$  & $0.2078$ & $-29.50$  \\
  & $4$ & $0.3377$ & $+14.55$  & $0.2550$ & $-13.51$  \\
  & $8$ & $0.2988$ & $+1.37$   & $0.2925$ & $-0.77$   \\
\bottomrule
\end{tabular*}

\vspace{2pt}
\parbox{\tblwidth}{%
  \footnotesize\normalfont\raggedright
  \emph{Note:} RD denotes relative deviation, computed with respect to the reference PoF of $0.2948$.%
}
\end{table}

\begin{table}[width=\linewidth,cols=7,pos=t]
\caption{Optimal bounds versus ITS sample size for the representative case \(K=8\) and \(r=2\) in the five-dimensional nonlinear smooth problem.}
\label{tab:ouq_nits_scan}
\begin{tabular*}{\tblwidth}{@{} L L L @{}}
\toprule
$N_{\mathrm{ITS}}$ & $U$ & $L$ \\
\midrule
$5\times10^3$   & 0.3058   & 0.2784    \\
$7\times10^3$   & 0.3061   & 0.2791    \\
$1\times10^4$   & 0.3056   & 0.2848    \\
$1.5\times10^4$ & 0.3062   & 0.2868    \\
$2\times10^4$   & 0.3041   & 0.2874    \\
$3.5\times10^4$ & 0.3029   & 0.2894  \\
$5\times10^4$   & 0.3031   & 0.2899   \\
$1\times10^5$   & 0.3014   & 0.2905   \\
\bottomrule
\end{tabular*}
\end{table}

\begin{table}[width=\linewidth,cols=6,pos=h]
\caption{Optimal bounds and relative deviations for the two-dimensional non-smooth four-branch problem with $Y^{\mathrm{c}}=0$.}
\label{tab:ouq_case3a_fourbranch}
\begin{tabular*}{\tblwidth}{@{} L L L L L L @{}}
\toprule
$r$ & $K$ & $U$ & RD of $U$ ($\%$) & $L$ & RD of $L$ ($\%$) \\
\midrule
\multirow{4}{*}{$0$}
  & $1$ & $1.0000$ & $+2.21\times10^4$ & $0.0000$ & $-100.00$ \\
  & $2$ & $1.0000$ & $+2.21\times10^4$ & $0.0000$ & $-100.00$ \\
  & $4$ & $1.0000$ & $+2.21\times10^4$ & $1.5423\times10^{-4}$ & $-96.58$ \\
  & $8$ & $0.0546$ & $+1.11\times10^3$ & $1.8928\times10^{-4}$ & $-95.80$ \\
\midrule
\multirow{4}{*}{$1$}
  & $1$ & $1.0000$ & $+2.21\times10^4$ & $0.0000$ & $-100.00$ \\
  & $2$ & $0.3404$ & $+7.45\times10^3$ & $0.0000$ & $-100.00$ \\
  & $4$ & $0.1402$ & $+3.01\times10^3$ & $1.5423\times10^{-4}$ & $-96.58$ \\
  & $8$ & $0.0149$ & $+229.53$ & $0.0015$ & $-66.37$ \\
\midrule
\multirow{4}{*}{$2$}
  & $1$ & $0.1145$ & $+2.44\times10^3$ & $0.0000$ & $-100.00$ \\
  & $2$ & $0.0433$ & $+860.90$ & $0.0000$ & $-100.00$ \\
  & $4$ & $0.0256$ & $+467.92$ & $2.2285\times10^{-4}$ & $-95.06$ \\
  & $8$ & $0.0081$ & $+79.38$ & $0.0023$ & $-49.31$ \\
\bottomrule
\end{tabular*}

\vspace{2pt}
\parbox{\tblwidth}{%
  \footnotesize\normalfont\raggedright
  \emph{Note:} RD denotes relative deviation, computed with respect to the reference PoF of $4.51\times10^{-3}$.
}
\end{table}

\begin{table}[width=\linewidth,cols=6,pos=h]
\caption{Optimal bounds and relative deviations for the two-dimensional non-smooth four-branch problem with $Y^{\mathrm{c}}=2$.}
\label{tab:ouq_case3b_fourbranch}
\begin{tabular*}{\tblwidth}{@{} L L L L L L @{}}
\toprule
$r$ & $K$ & $U$ & RD of $U$ ($\%$) & $L$ & RD of $L$ ($\%$) \\
\midrule
\multirow{4}{*}{$0$}
  & $1$ & $1.0000$ & $+8.62\times10^6$ & $0.0000$ & $-100.00$ \\
  & $2$ & $1.0000$ & $+8.62\times10^6$ & $0.0000$ & $-100.00$ \\
  & $4$ & $0.0124$ & $+1.07\times10^5$ & $0.0000$ & $-100.00$ \\
  & $8$ & $0.0014$ & $+1.20\times10^4$ & $2.1889\times10^{-6}$ & $-81.13$ \\
\midrule
\multirow{4}{*}{$1$}
  & $1$ & $1.0000$ & $+8.62\times10^6$ & $0.0000$ & $-100.00$ \\
  & $2$ & $0.0507$ & $+4.37\times10^5$ & $0.0000$ & $-100.00$ \\
  & $4$ & $9.5711\times10^{-4}$ & $+8.15\times10^3$ & $0.0000$ & $-100.00$ \\
  & $8$ & $2.4623\times10^{-4}$ & $+2.02\times10^3$ & $2.1889\times10^{-6}$ & $-81.13$ \\
\midrule
\multirow{4}{*}{$2$}
  & $1$ & $0.0151$ & $+1.30\times10^5$ & $0.0000$ & $-100.00$ \\
  & $2$ & $0.0066$ & $+5.64\times10^4$ & $0.0000$ & $-100.00$ \\
  & $4$ & $1.3068\times10^{-4}$ & $+1.03\times10^3$ & $0.0000$ & $-100.00$ \\
  & $8$ & $4.0035\times10^{-5}$ & $+245.13$ & $2.2345\times10^{-6}$ & $-80.74$ \\
\bottomrule
\end{tabular*}

\vspace{2pt}
\parbox{\tblwidth}{%
  \footnotesize\normalfont\raggedright
  \emph{Note:} RD denotes relative deviation, computed with respect to the reference PoF of $1.16\times10^{-5}$.%
}
\end{table}

\begin{table}[width=\linewidth,cols=6,pos=h]
\caption{Optimal bounds and relative deviations for the rare-event roof-truss problem.}
\label{tab:ouq_case4_rooftruss_bounds}
\begin{tabular*}{\tblwidth}{@{} L L L L L L @{}}
\toprule
$r$ & $K$ & $U$ & RD of $U$ ($\%$) & $L$ & RD of $L$ ($\%$) \\
\midrule
\multirow{4}{*}{$0$}
  & $1$ & $1.0000$             & $+1.99\times10^{8}$ & $0.0000$ & $-100.00$ \\
  & $2$ & $0.3750$             & $+7.46\times10^{7}$ & $0.0000$ & $-100.00$ \\
  & $4$ & $0.1250$             & $+2.49\times10^{7}$ & $0.0000$ & $-100.00$ \\
  & $8$ & $0.0334$             & $+6.66\times10^{6}$ & $0.0000$ & $-100.00$ \\
\midrule
\multirow{4}{*}{$1$}
  & $1$ & $0.4976$             & $+9.90\times10^{7}$ & $0.0000$ & $-100.00$ \\
  & $2$ & $0.0255$             & $+5.07\times10^{6}$ & $0.0000$ & $-100.00$ \\
  & $4$ & $0.0042$             & $+8.43\times10^{5}$ & $0.0000$ & $-100.00$ \\
  & $8$ & $9.0931\times10^{-4}$ & $+1.81\times10^{5}$ & $0.0000$ & $-100.00$ \\
\midrule
\multirow{4}{*}{$2$}
  & $1$ & $0.0157$             & $+3.12\times10^{6}$ & $0.0000$ & $-100.00$ \\
  & $2$ & $1.1072\times10^{-3}$ & $+2.20\times10^{5}$ & $0.0000$ & $-100.00$ \\
  & $4$ & $2.6038\times10^{-4}$ & $+5.17\times10^{4}$ & $0.0000$ & $-100.00$ \\
  & $8$ & $8.6255\times10^{-5}$ & $+1.71\times10^{4}$ & $0.0000$ & $-100.00$ \\
\bottomrule
\end{tabular*}

\vspace{2pt}
\parbox{\tblwidth}{%
  \footnotesize\normalfont\raggedright
  \emph{Note:} RD denotes relative deviation, computed with respect to the reference PoF of $5.0249\times10^{-7}$.
}
\end{table}

\begin{table}[width=\linewidth,cols=6,pos=h]
\caption{Optimal bounds and relative deviations for the ten-dimensional ballistic impact problem.}
\label{tab:ouq_10d_impact_bounds}
\begin{tabular*}{\tblwidth}{@{} L L L L L L @{}}
\toprule
$r$ & $K$ & $U$ & RD of $U$ ($\%$) & $L$ & RD of $L$ ($\%$) \\
\midrule
\multirow{4}{*}{$0$}
  & $1$ & $1.0000$ & $+7389.51$ & $0.0000$ & $-100.00$ \\
  & $2$ & $0.3359$ & $+2416.01$ & $0.0000$ & $-100.00$ \\
  & $4$ & $0.1144$ & $+757.10$  & $0.0000$ & $-100.00$ \\
  & $8$ & $0.0485$ & $+263.39$  & $2.68\times10^{-4}$ & $-97.99$ \\
\midrule
\multirow{4}{*}{$1$}
  & $1$ & $0.3217$ & $+2309.47$ & $0.0000$ & $-100.00$ \\
  & $2$ & $0.0715$ & $+435.14$ & $0.0000$ & $-100.00$ \\
  & $4$ & $0.0321$ & $+140.74$ & $0.0024$ & $-81.92$ \\
  & $8$ & $0.0170$ & $+27.35$  & $0.0117$ & $-12.52$ \\
\midrule
\multirow{4}{*}{$2$}
  & $1$ & $0.0892$ & $+568.03$ & $0.0000$ & $-100.00$ \\
  & $2$ & $0.0329$ & $+146.48$ & $4.0\times10^{-6}$ & $-99.97$ \\
  & $4$ & $0.0152$ & $+13.66$  & $0.0110$ & $-17.39$ \\
  & $8$ & $0.0136$ & $+1.75$   & $0.0131$ & $-1.62$ \\
\bottomrule
\end{tabular*}

\vspace{2pt}
\parbox{\tblwidth}{%
  \footnotesize\normalfont\raggedright
  \emph{Note:} RD denotes relative deviation, computed with respect to the reference PoF of $0.013352$.%
}
\end{table}

\begin{table}[width=\linewidth,cols=6,pos=h]

\caption{Optimal bounds and relative deviations for the ten-dimensional ballistic impact problem under varying failure threshold $Y^{\mathrm{c}}$.}
\label{tab:ouq_10d_impact_yc}
\begin{tabular*}{\tblwidth}{@{} L L L L L L @{}}
\toprule
$Y^{\mathrm{c}}$ (cm) & Reference PoF & $U$ & RD of $U$ ($\%$) & $L$ & RD of $L$ ($\%$) \\
\midrule
$0.8993$ & $0.0999$ & $0.1181$ & $+18.23$ & $0.0915$ & $-8.36$ \\
$0.9060$ & $0.0723$ & $0.0915$ & $+26.55$ & $0.0619$ & $-14.35$ \\
$0.9128$ & $0.0500$ & $0.0701$ & $+40.19$ & $0.0389$ & $-22.18$ \\
$0.9230$ & $0.0253$ & $0.0459$ & $+81.83$ & $0.0119$ & $-52.76$ \\
$0.9327$ & $0.0099$ & $0.0267$ & $+170.09$ & $8.72\times10^{-4}$ & $-91.17$ \\
$0.9380$ & $0.0050$ & $0.0166$  & $+232.07$ & $4.4\times10^{-5}$ & $-99.12$ \\
$0.9460$ & $9.86\times10^{-4}$ & $0.0067$ & $+576.67$ & $4.0\times10^{-6}$ & $-99.59$ \\
\bottomrule
\end{tabular*}

\vspace{2pt}
\parbox{\tblwidth}{%
  \footnotesize\normalfont\raggedright
  \emph{Note:} RD denotes relative deviation, computed with respect to the corresponding reference PoF at each $Y^{\mathrm{c}}$.%
}
\end{table}

\printcredits

\bibliographystyle{model1-num-names}

\bibliography{cas-refs}

\end{document}